\documentclass{article}

\usepackage{arxiv}
\usepackage[utf8]{inputenc} 
\usepackage[T1]{fontenc}    
\usepackage{hyperref}       
\usepackage{url}            
\usepackage{booktabs}       
\usepackage{amsfonts}       
\usepackage{nicefrac}       
\usepackage{microtype}      
\usepackage{lipsum}		
\usepackage{graphicx}
\usepackage[square,sort,comma,numbers]{natbib}
\usepackage{doi}
\usepackage{stackrel}
\usepackage{accents}
\usepackage{wrapfig}
\usepackage{algorithm}
\usepackage[noend]{algpseudocode}
\usepackage[bottom]{footmisc}
\newtheorem{theorem}{Theorem}
\newtheorem{lemma}{Lemma}
\newtheorem{corollary}{Corollary}

\newenvironment{proof}{\paragraph{Proof.}}{\hfill$\square$}

\title{On the Sampling Sparsity of Neuromorphic Analog-to-Spike Conversion based on Leaky-Integrate-and-Fire}
\author{ 
{
\hspace{1mm}Bernhard A.~Moser}\thanks{double affiliation: Software Competence Center Hagenberg (SCCH), 4232 Hagenberg, Austria} \\
	Institute of Signal Processing \\
	Johannes Kepler University of Linz\\
	\texttt{bernhard.moser@\{scch.at,jku.at\}} 
	\And
	{\hspace{1mm}Michael Lunglmayr} \\
	Institute of Signal Processing\\
	Johannes Kepler University of Linz, Austria\\
	\texttt{michael.lunglmayr@jku.at} 
	}

\renewcommand{\shorttitle}{\textit{arXiv} Sampling Sparsity of Analog-to-Spike Conversion}

\begin{document}
\maketitle

\begin{abstract}
In contrast to the traditional principle of periodic sensing neuromorphic engineering 
pursues a paradigm shift towards bio-inspired event-based sensing, where events are primarily 
triggered by a change in the perceived stimulus. 
We show in a rigorous mathematical way that information encoding by means of Threshold-Based Representation  
based on either Leaky Integrate-and-Fire (LIF) or Send-on-Delta (SOD) is linked to an analog-to-spike conversion that guarantees maximum sparsity while satisfying an approximation condition based on the Alexiewicz norm.

\end{abstract}


\section{Introduction}
In biological sensing systems, the information encoding of analog signals 
is driven by event-based mechanisms~\cite{Tayarani-Najaran2021}.
The unrivaled efficiency of biological information processing has long motivated engineers to seek bio-inspired approaches to sensing and signal processing. The main principle of neuromorphic sensing is to replace the traditional approach of periodic sensing in favor of an event-driven scheme that mimics sensing in the nervous system, where events are primarily triggered by a change in the perceived stimulus.
This way of temporal information encoding is also called Threshold-Based Representation (TBR) in the literature~\cite{Forno2022}, 
for which threshold-based mechanisms based on Leaky Integrate-and-Fire (LIF) and Send-on-Delta (SOD) provide computational models. Nowadays LIF has become popular in neuromorphic computing in the context of spiking neural networks (SNNs)~\cite{bookGerstner2014,TAVANAEI2019,Nunes2022,Yamazaki2022}.
TBR is applied successfully in neuromorphic sensing and SNN applications based on time-varying signals. For example, for the Free Spoken Digit (FSD) benchmark dataset, consisting of $8$-kHz audio files, it was shown that this type of temporal contrast coding outperforms other coding variants such as rate coding or deconvolution-based approaches~\cite{Forno2022}. Other applications of TBR can be found in neuromorphic visual processing applications~\cite{Yousefzadeh2022}.

In this paper, we provide a mathematical framework that allows us to  show that the information encoding mechanisms of TBR based on 
LIF and SOD distinguish by maximal sparsity while meeting certain approximation conditions. An important contribution of this paper is the identification of these approximation conditions, which is closely related to the Alexiewicz norm, a norm rarely used in mathematics and engineering so far~\cite{Alexiewicz1948, MOSER2024128190}.

The paper is outlined as follows. 
In Section~\ref{s:LIF} we fix  notation and provide a mathematical definition of the LIF model in continuous time, 
comprising the parameters for the refractory time, the leakage and the threshold. 
As a starter, in Section~\ref{s:IF} we motivate our analysis by considering the special case of zero leakage, i.e. integrate-and-fire (IF),
which allows a mathematical representation $\mbox{IF} = A^{-1} \circ q \circ A$ in terms of a transformed version of the standard element-wise quantization 
operation $q$ by rounding to the next integer closest to zero. 
This decomposition theorem reveals the related metric structure of the underlying signal space in a direct way, which turns out to be the Alexiewicz norm~\cite{Alexiewicz1948, MOSER2024128190}. 
Since spike trains are discrete sequences of spikes, their number, respectively, the sum over all spikes and therefore its $L_1$-norm is a natural measure 
for the sparseness of a spike train. In this way, in Section~\ref{ss:zero} we 
investigate sparseness, and, in a first step, prove an extremal sparseness property for IF by exploiting  the outlined decomposition theorem of Section~\ref{s:IF}. 
In Section~\ref{s:LIF} we consider the general situation of arbitrary leakage. 
The fact that the decomposition property of IF fails in the general case is an indication that the extremal sparseness property does not generally apply. However, experimental evaluations show that LIF with any  non-zero leakage parameter has a very high probability of exhibiting extremal sparseness. 

\section{Mathematical Preliminaries}
\label{s:LIF}
In the biological nervous system, information processing relies on event-based triggering mechanisms.
Due to the principle of Temporal Encoding and its realization by Threshold-Based Representation (TBR), for sensing purposes a simple idealized mathematical model of such a mechanism can be realized by means of triggering (or firing) a polar pulse (i.e., spike with either positive or negative sign) when a certain threshold value for the signal's change is exceeded. In case the change is positive, a positive spike is created, if the change is negative, 
a negative spike is created~\cite{Forno2022}. Depending on how {\it changes} are modeled we obtain variants of thresholding-based representations. For SOD, the signal's change is measured by the absolute difference 
$|f(t_{k+1}) -  f(t_k)|$ between the signal values $f(t_k)$ at the time $t_k$ of the last event and 
the current value $f(t_{k+1})$. Positive, resp. negative Dirac impulses, also referred to as spikes, are triggered
if $f(t_{k+1}) -  f(t_k)\geq \vartheta$, resp. $f(t_{k+1}) -  f(t_k)\leq -\vartheta$, for some given threshold $\vartheta$.
In contrast to SOD, the LIF model employs a temporal aggregation by means of a leaky integral. Later on in Section~\ref{ss:zero} we point out that SOD can be traced back to LIF by considering the derivative of the analog signal. 
Therefore let us first give a definition of LIF for generating positive or negative spikes depending on the sign of the signal's change.
For this we consider the leaky parameter $\alpha\geq 0$, the threshold $\vartheta>0$ and assume an {\it a.e. bounded} signal $f$ with locally finitely many Dirac impulses. 
LIF can be understood as a mapping that converts $f$ into a spike train $s(t) = \sum_k s_k\delta(t - t_k)$, where $s_k \in \mathbb{R}$ denotes the amplitude of the spike at time $t_k$. 
The time points $t_{k}$ are recursively given by
\begin{equation}
\label{eq:LIFsample}
t_{k+1} := \inf\left\{T\geq t_k + t_r: \left|\int_{t_k}^{T} e^{-\alpha (T -  t)} \big(f(t) + r_k \delta(t-t_k)\big) dt\right| \geq \vartheta\right\},
\end{equation}
where $t_r\geq 0$ is the refractory time and $T=t_{k+1}$ is the first time point after $t_{k}$ that causes the integral in 
(\ref{eq:LIFsample}) to violate the sub-threshold condition $|\int_{t_k}^{T} e^{-\alpha (t_{k+1} -  t)} f(t) dt | < \vartheta$.
The term $r_k \delta(t-t_k)$ refers to the reset of the membrane potential in the moment a spike has been triggered.
In the standard definition of LIF for discrete spike trains, see~\cite{bookGerstner2014,lapicque_recherches_1907}, the reset is defined as the membrane potential that results from subtracting the threshold if the membrane's potential reaches the positive threshold level $+\vartheta$, respectively adding $\vartheta$ if 
the membrane's potential reaches the negative threshold level $-\vartheta$. 
In the case of bounded $f$ the integral $g(t):= \int_{t_k}^t e^{-\alpha (t_{k+1} -  t)} f(t) dt$ is changing continuously in $t$ 
so that the threshold level in (\ref{eq:LIFsample}) is exactly hit. 
Consequently the resulting reset amounts to zero, i.e., $r_k = 0$ and the resulting amplitude $s_k$ of the triggered spike is defined accordingly, i.e., $s_k = +\vartheta$, when the positive threshold value is reached, and $s_k = -\vartheta$ when the negative threshold value is reached.
For a mathematical analysis and a discussion of how to define the reset $r_k$  in the presence of Dirac impulses see~\cite{MOSER2024128190}.

In our analysis we choose the subtraction reset mechanism in the setting of zero refractory time, i.e, $t_r=0$, which equals the {\it reset-to-mod} 
variant outlined in~\cite{MOSER2024128190}.
Now, by fixing the leakage parameter $\alpha\geq 0$ and the threshold $\vartheta>0$, 
Equ.~(\ref{eq:LIFsample}) turns a signal $f \in \mathcal{F}$ from the signal space $\mathcal{F}$ under consideration
into a sequence of spikes $s = (t_k, s_k) \in \mathbb{S}_{\vartheta}$, establishing the
mapping 
\begin{equation}
\label{eq:DefLIF} 
\mbox{LIF}_{\alpha, \vartheta}(f)= s \in \mathbb{S}_{\vartheta}, 
\end{equation}
where $\mathbb{S}_{\vartheta}$ denotes the vector space of spike trains $s$ represented in terms of a
sequence $(t_k, s_k)_k$ or by means of a sum of Dirac pulses, $s(t) = \sum_k s_k \, \delta(t - t_k)$, with
spike amplitudes being threshold multiples, $s_k \in \vartheta\,\mathbb{Z}$, and 
the time points $(t_k)_k$ being at most finitely many in any compact time interval, see~\cite{MoserLunglmayr_AIROV2024}. 
As pointed out in~\cite{MoserLunglmayr_AIROV2024}, Equ. (\ref{eq:DefLIF}) is well-defined for any locally integrable and almost everywhere bounded $f$ with locally finitely many Dirac impulses.
The special case of zero leakage, $\alpha=0$, is also referred to as integrate-and-fire (IF), i.e., $\mbox{IF}_{\vartheta}(f) = \mbox{LIF}_{0, \vartheta}(f)$.

The assumption of zero refractory time, $t_r=0$, turns the recursive scheme~(\ref{eq:LIFsample}) to the following recursion:
\begin{eqnarray}
\label{eq:DefLIFa}
z_{t_1} & := & F_{t_1}, 																	\nonumber \\
s_{t_1} & := & q_{\vartheta}(z_{t_1}), 															\nonumber \\
z_{t_{k+1}} & := & F_{t_{k+1}} + e^{-\alpha\, (t_{k+1} - t_k)}\, (z_{t_k} - s_{t_k}),		\nonumber \\
s_{t_{k+1}} & := & q_{\vartheta}(z_{t_{k+1}}),
\end{eqnarray}
where $F_{t_{k+1}}:= \int_{(t_k, t_{k+1}]} e^{-\alpha\, (t_{k+1} -t)} f(t) dt$ is the evaluation of the integral immediately after 
the spike triggering event at $t = t_k$.  
$q_{\vartheta}(z):= \mbox{sgn}(z)\, \vartheta\, \max\{n \in \mathbb{N}_0: \vartheta\, n \leq |z|\}$ is the standard quantization by truncation, see Fig.~\ref{fig:q}, 
where $\mbox{sgn}(z) = 1$ if $z>0$, $-1$ if $z<0$ and $0$ else.
$t_{k+1}$ is found to be the first time $T> t_k$ for which $|F_T + e^{-\alpha\, (T - t_k)}\, (z_{k} - s_{k})| \geq \vartheta$.
\begin{figure}
	\centering
	\includegraphics[width=0.3\linewidth]{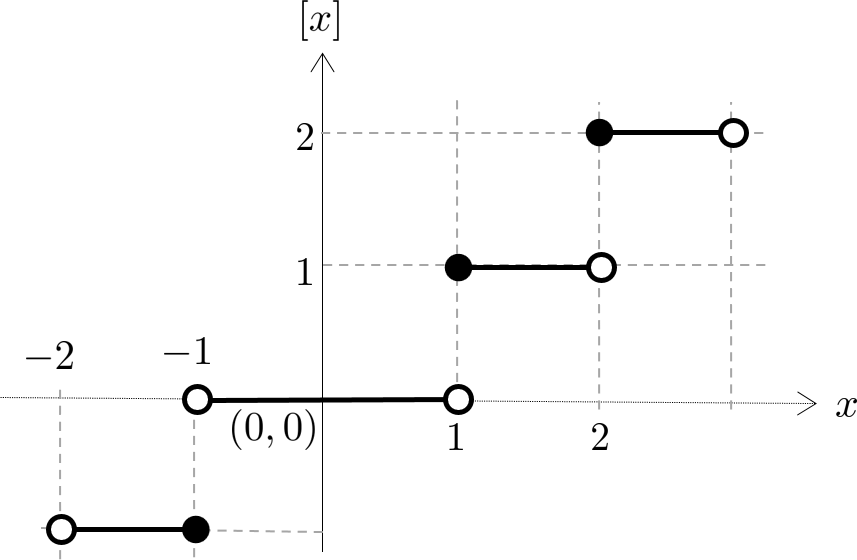}
	\caption{Illustration of standard quantization function $q(z):=q_{\vartheta}(z)$ for $\vartheta = 1$}
	\label{fig:q}
\end{figure}

In the context of artificial spiking neural networks (SNNs)  continuous time is often replaced by discrete time at $t_k := k\,\Delta$ and the signal 
$f = f(t)$ becomes a sequence $f_k := f(t_k)\, 1_{(k-1, k]}$. 
This way, (\ref{eq:DefLIFa}) simplifies to the recursion ($k = 1,2, \ldots$):
\begin{eqnarray}
\label{eq:DefLIFd}
z_1 & := & f_1, 																	\nonumber \\
s_1 & := & q_{\vartheta}(z_1), 															\nonumber \\
z_{k+1} & := & f_{k+1} + \beta\, (z_k - s_k),		\nonumber \\
s_{k+1} & := & q_{\vartheta}(z_{k+1}),
\end{eqnarray}
where $\beta := e^{-\alpha} \in [0,1]$.
 
To unify notation we write $\mathcal{F}$ for the function space under consideration. In the discrete case $\mathcal{F}$ is the set of all sequences 
$f = (f_1, f_2, \ldots)$, $f_k \in \mathbb{R}$, and in the continuous-time case, $\mathcal{F}$ denotes the space of bounded integrable functions 
$f: [0, \infty) \rightarrow \mathbb{R}$ superimposed at most by finitely many Dirac impulses.

\section{Spike-Train Quantization in the Alexiewicz Norm}
\label{s:IF}
In~\cite{MOSER2024128190} we show that LIF shows characteristics of a quantization operator in terms of 
an idempotent projection, $\mbox{LIF}_{\alpha, \vartheta}(f) = \mbox{LIF}_{\alpha, \vartheta}(\mbox{LIF}_{\alpha, \vartheta}(f)) = \mbox{LIF}_{\alpha, \vartheta}(f)$, and a quantization error bound given by 
\begin{equation}
\label{eq:LIFQuant}
\|\mbox{LIF}_{\alpha, \vartheta}(f) - f\|_{A, \alpha} < \vartheta, 
\end{equation}
where $\|.\|_{A, \alpha}$ is the weighted Alexiewicz norm $\|f\|_{A, \alpha} = \sup_T|\int_0^T e^{-\alpha (T-t)} f(t) dt|$, and 
$f$ is locally integrable and bounded except a finite number of Dirac impulses~\cite{MoserLunglmayr_AIROV2024}.
For the discrete case of equidistant unit time steps  we get $\|f\|_{A, \alpha} = \max_n |\sum_{k = 0}^n  \beta^{n-k}  f_k|$, where $f(t)= \sum_k f_k \delta(t-k)$ and $\beta := e^{-\alpha}$. 

The open Alexiewicz ball $\mathring{B}_{\alpha, \vartheta}(f)$ for some $\alpha\in [0, \infty]$ with radius $\vartheta>0$ centered at $f$ is given by
$\mathring{B}_{\alpha, \vartheta}(f) := \{g :\, \|f - g\|_{A, \alpha} < \vartheta\}$. Note that for $f\equiv 0$ the open Alexiewicz ball is just the set of all 
sub-threshold functions $f$ for which never a spike is triggered.

\subsection{Geometry of the Alexiewicz Ball}
\label{s:AlexBall}
In our mathematical framework the Alexiewicz norm is key, so let us study its induced geometry.
As illustrated in Fig.~\ref{fig:AlexiewiczBall} in contrast to standard norms the unit ball induced by the Alexiewicz norm is less symmetric.
\begin{figure}[ht]
	\centering
	\includegraphics[width=0.3\linewidth]{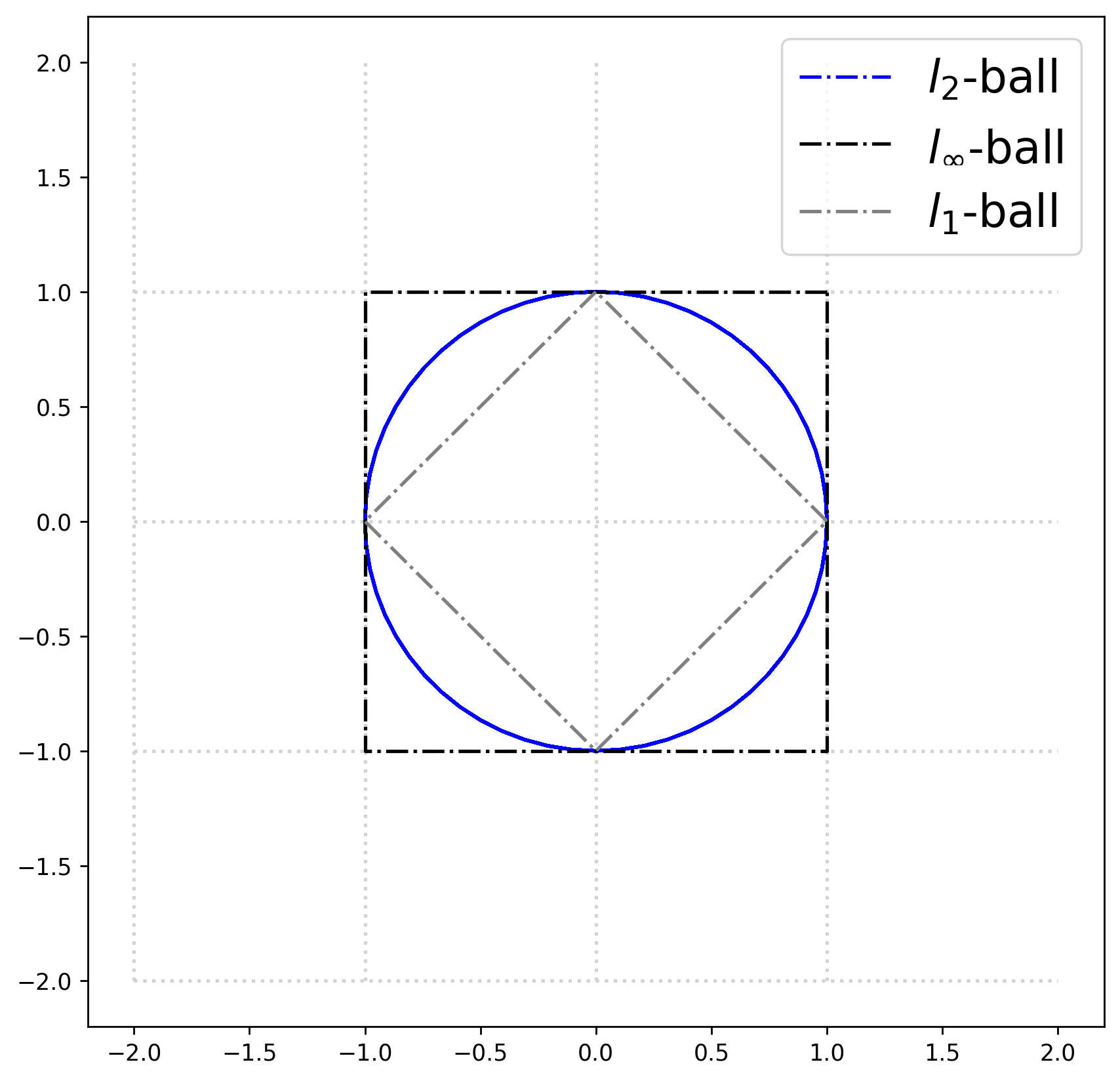}
	\includegraphics[width=0.3\linewidth]{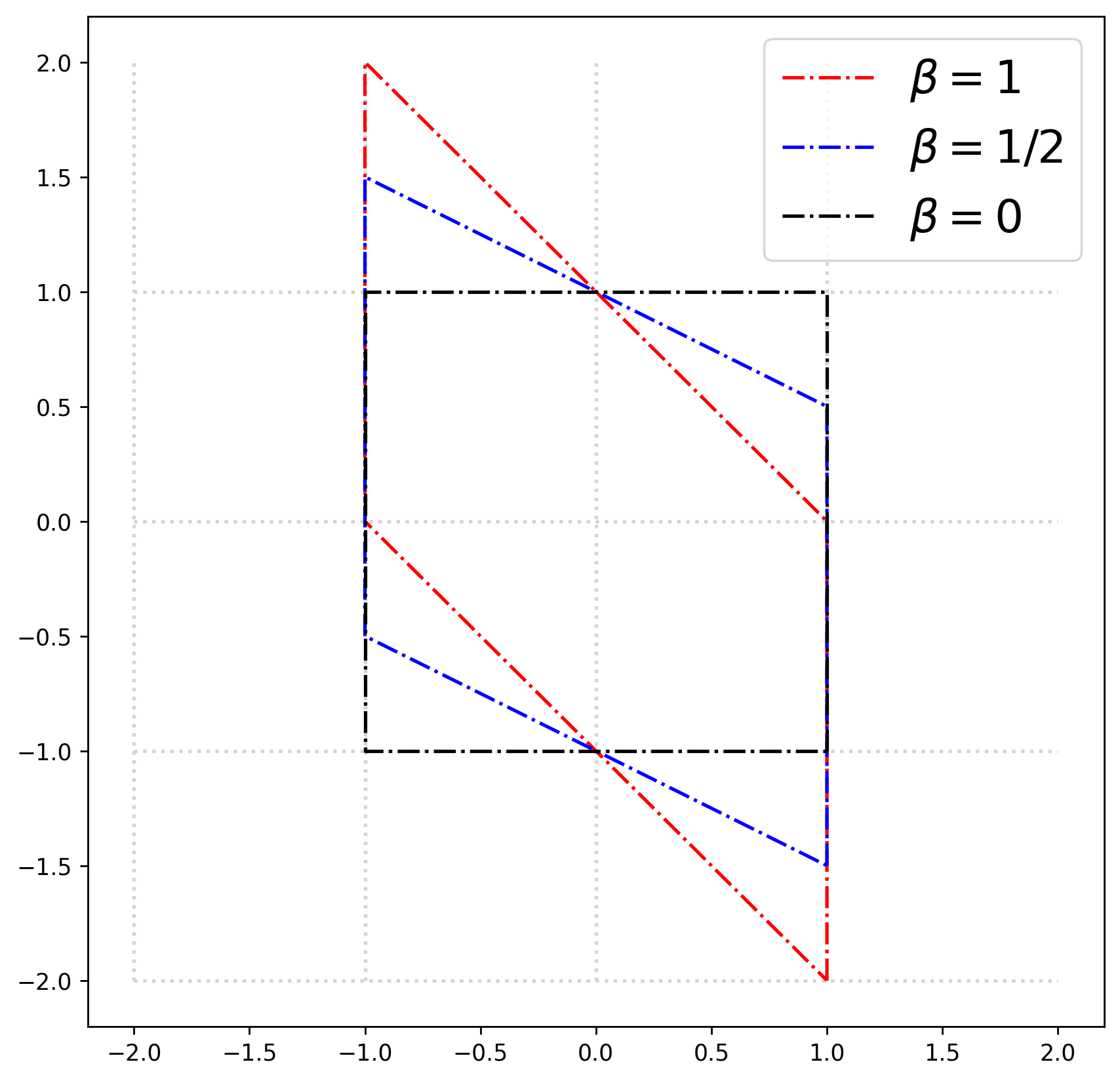}
		\caption{In contrast to the standard $l_p$ norms (left),  for the Alexiewicz nom $\|.\|_{A, \alpha}$, $\beta = e^{-\alpha}$, time ordering matters 
		resulting in a less symmetric geometry of its unit ball (right). 
		}
		\label{fig:AlexiewiczBall}
\end{figure}
This reflects the fact that for the  Alexiewicz norm time-ordering matters. Consider for example
$\|(1, -1, 1)\|_{A, \alpha} = 1 \neq 5/4 = \|(-1, 1, 1)\|_{A, \alpha}$ for $\alpha = \ln(2)$, i.e., $\beta = e^{-\alpha} = 1/2$.

The  Alexiewicz ball can be represented by means of a sheared transformation of the maximum norm ball
\begin{lemma}
\label{lem:ball}
\begin{equation}
\label{eq:ball}
\mathring{B}_{\alpha, \vartheta}(f) = f + A_{\alpha}^{-1}\circ \mathring{B}_{\infty, \vartheta}(0),
\end{equation}
where  
\begin{equation}
\label{eq:Acont}
\hat{f}(t):= A_{\alpha}(f)(t):= \int_0^t e^{-\alpha\, (t-\tau)} f(\tau) d\tau, 
\end{equation}
for which the inverse is found (by using the product rule) to be $f(t) = A_{\alpha}^{-1}(\hat{f}):=  \alpha\, \hat{f}(t) + \frac{d}{dt}\hat{f}(t)$,
where  $\frac{d}{dt}\hat{f}(t)$ denotes the distributional derivative of $\hat{f}$.
\end{lemma}
\begin{proof}
$\mathring{B}_{\alpha, \vartheta}(f)  =  \{f + h: \|h\|_{A, \alpha} < \vartheta\} 
																		  =  \{f + h: \|A_{\alpha}\, h\|_{\infty} < \vartheta\} 
																		  =  \{f + A_{\alpha}^{-1}\, \widetilde{h}: \|\widetilde{h}\|_{\infty} < \vartheta\}
																		 = f + A_{\alpha}^{-1}\circ \mathring{B}_{\infty, \vartheta}(0)$.
\end{proof}

For examples of Alexiewicz balls with $2$ time points see Fig.~\ref{fig:2DLIF}.
\begin{figure}[ht]
	\centering
	\includegraphics[width=0.5\linewidth]{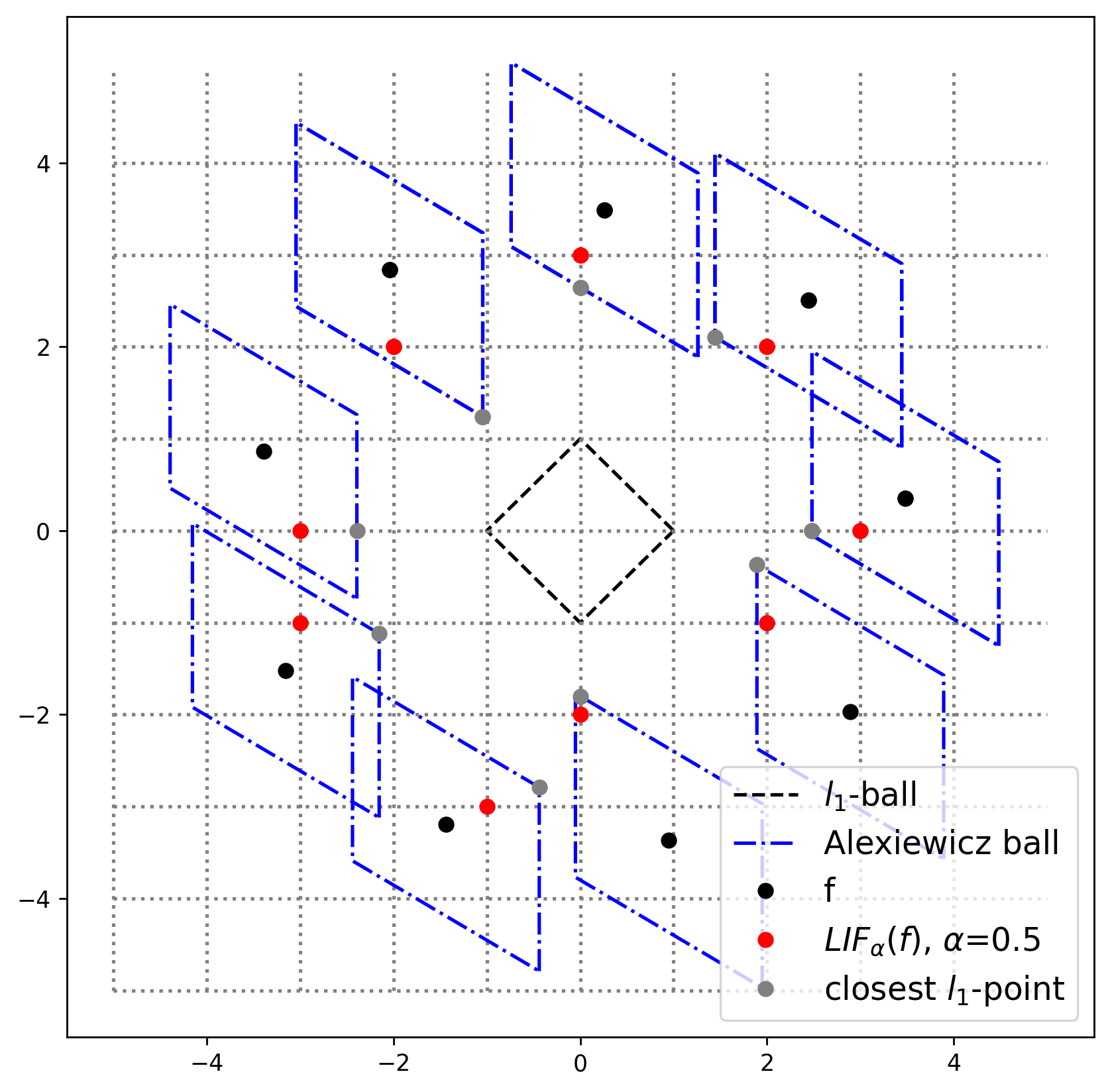}
		\caption{Illustration of spike trains $s = s_1\, \delta(t-t_1) + s_2\, \delta(t-t_2)$ consisting of just $2$ spikes  $(s_1, s_2)$, represented as 2D points.
		$l_1$-ball in the center (black dashed line) and 
				2D unit Alexiewicz balls (blue) centered at input signals $f$ (black points), with closest $l_1$-points (gray points) and LIF spike train (red points).
				In this setting LIF satisfies the extremal sparsity property in the sense that the red points are those grid points within the Alexiewics balls that have
				minimal $l_1$-norm.
		}
			\label{fig:2DLIF}
\end{figure}

\subsection{Sparsity Bounds for Leaky Integrate-and-Fire}
\label{ss:l1Distance}
For a spike train $s = \sum_k s_k\, \delta(t -  t_k)$,  the notion of sparsity is equivalent with its $l_1$-norm.
Strictly speaking, sparsity would be defined via the number of non-zero elements. 
As this number does not resemble a norm, the $l_1$-norm norm is used instead as it is well established in literature that this norm fosters sparsity~\cite{Ramirez2013}.
We are interested in lower and upper bounds on the sparsity of the spike train $s$ obtained by leaky integrate-and-fire applied on some $f$.
As a first observation, due to~(\ref{eq:LIFQuant}),  we get a lower bound by taking into account the fact that LIF maps $f$ into the Alexiewicz ball centered at $f$, 
providing $\|B_{\alpha, \vartheta}(f)\|_1 := \inf\{\|g\|_1:  \|f - g\|_{A, \alpha} < \vartheta\} \leq \|\mbox{LIF}_{\alpha, \vartheta}(f)\|_1$. 
On the other hand, LIF is sparse in the sense that the $l_1$-norm of the resulting spike train is below that of $f$. 
\begin{theorem}
\label{prop:l1inequ}
For any $f$  that is  integrable and bounded except a finite number of Dirac impulses
leaky integrate-and-fire satisfies the following sparsity bounds
\begin{equation}
\label{eq:l1inequ}
 \|B_{\alpha, \vartheta}(f)\|_1 \leq \|\mbox{LIF}_{\alpha, \vartheta}(f)\|_1 \leq  \|f\|_1.
\end{equation}
\end{theorem}

\begin{proof}
Here we show the proof for the upper bound in the discrete case, $f = (f_k)_k$. 
By applying~(\ref{eq:DefLIFa}) the proof for the time-continuous case can be handled in an analogous way, see Appendix A.

Note that for any $x$ we have $|x| = |q(x)| + |x - q(x)|$. 
This observation allows us to restructure the $l_1$-norm of the spike train obtained by $\mbox{LIF}_{\alpha, \vartheta}(f)$ into the following telescope sum.
Sticking to the notation of (\ref{eq:DefLIFd}) we obtain
\noindent
\begin{eqnarray}
&  \|\mbox{LIF}_{\alpha, \vartheta}(f)\|_1 \nonumber \\
& =   \sum_k |s_k| \nonumber \\
& =  \sum_k |q_{\vartheta}(f_k + \beta\, (z_{k-1} - s_{k-1}))| \nonumber \\
						 & =  \sum_k |f_k + \beta\, (z_{k-1} - s_{k-1})|   \nonumber \\
						 & \,\,\,\,\,\,\,\, - |f_k + \beta\, (z_{k-1} - s_{k-1}) - q_{\vartheta}(f_k + \beta\, (z_{k-1} - s_{k-1}))| \nonumber \\
						 & \leq  \sum_k |f_k| + \beta\, |z_{k-1} - s_{k-1}| -  \beta\,|\underbrace{f_k + \beta\, (z_{k-1} - s_{k-1})}_{z_{k}} - \underbrace{q_{\vartheta}(z_k)}_{s_k}|
						\nonumber \\
						 & \leq  \sum_k |f_k| = \|f\|_1. \nonumber	
\end{eqnarray}
\end{proof}

As outlined in the Appendix B, the lower bound in (\ref{eq:l1inequ}) can be computed by solving a constrained $l_1$-minimization problem which in this case can be resolved efficiently by a simple recursion. Fig.~\ref{fig:ratio1}-\ref{fig:ratio2} show how the sparsity of LIF is distributed between these bounds.
If there are no (or only weak) constraints on $f$ there is less difference between the distributions for different leakage parameters, and 
the mean sparsity of LIF tends to the arithmetic mean of the bounds. 
\begin{figure}[ht]
	\centering
	\includegraphics[width=0.45\linewidth]{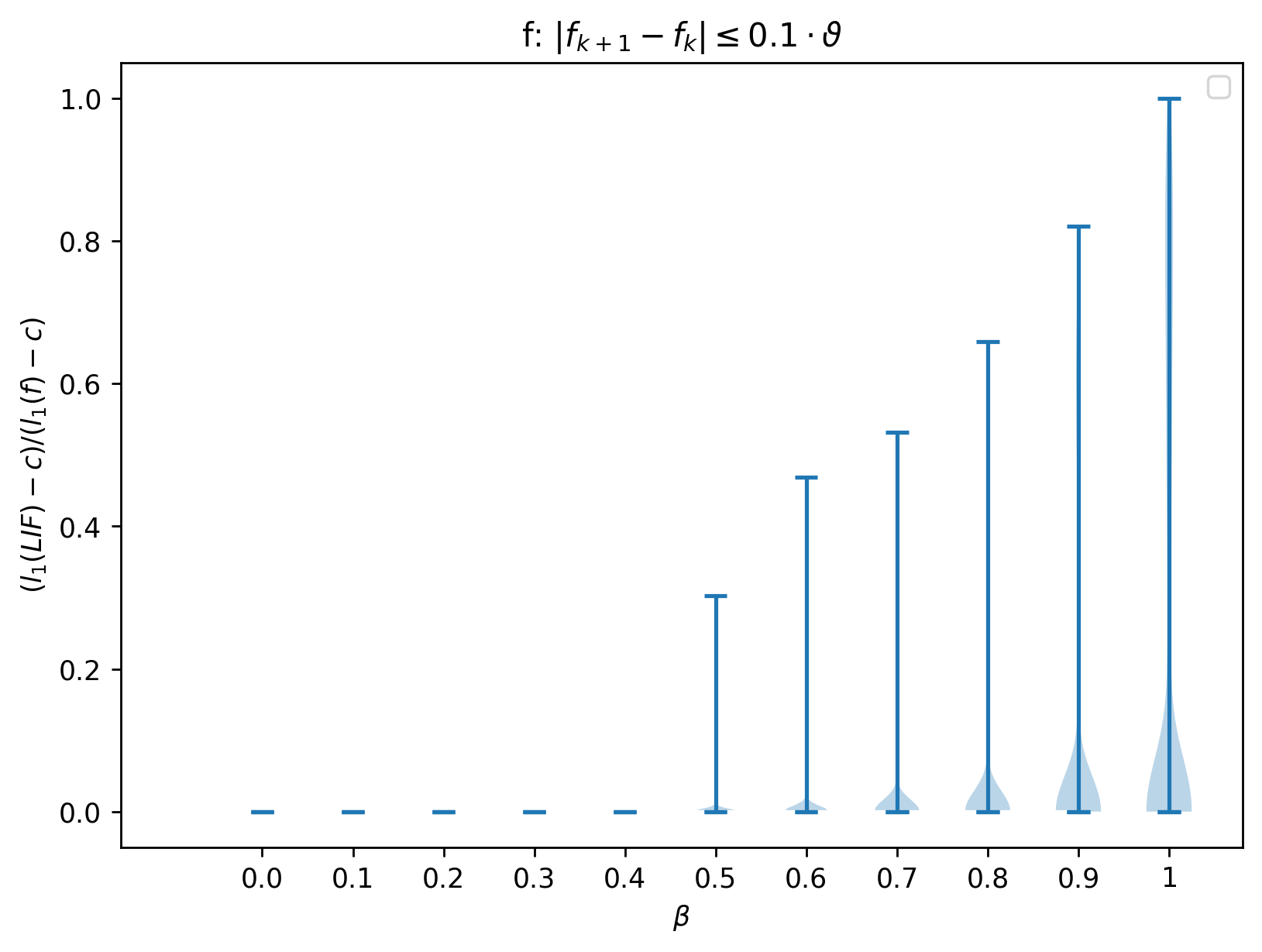}
	\includegraphics[width=0.45\linewidth]{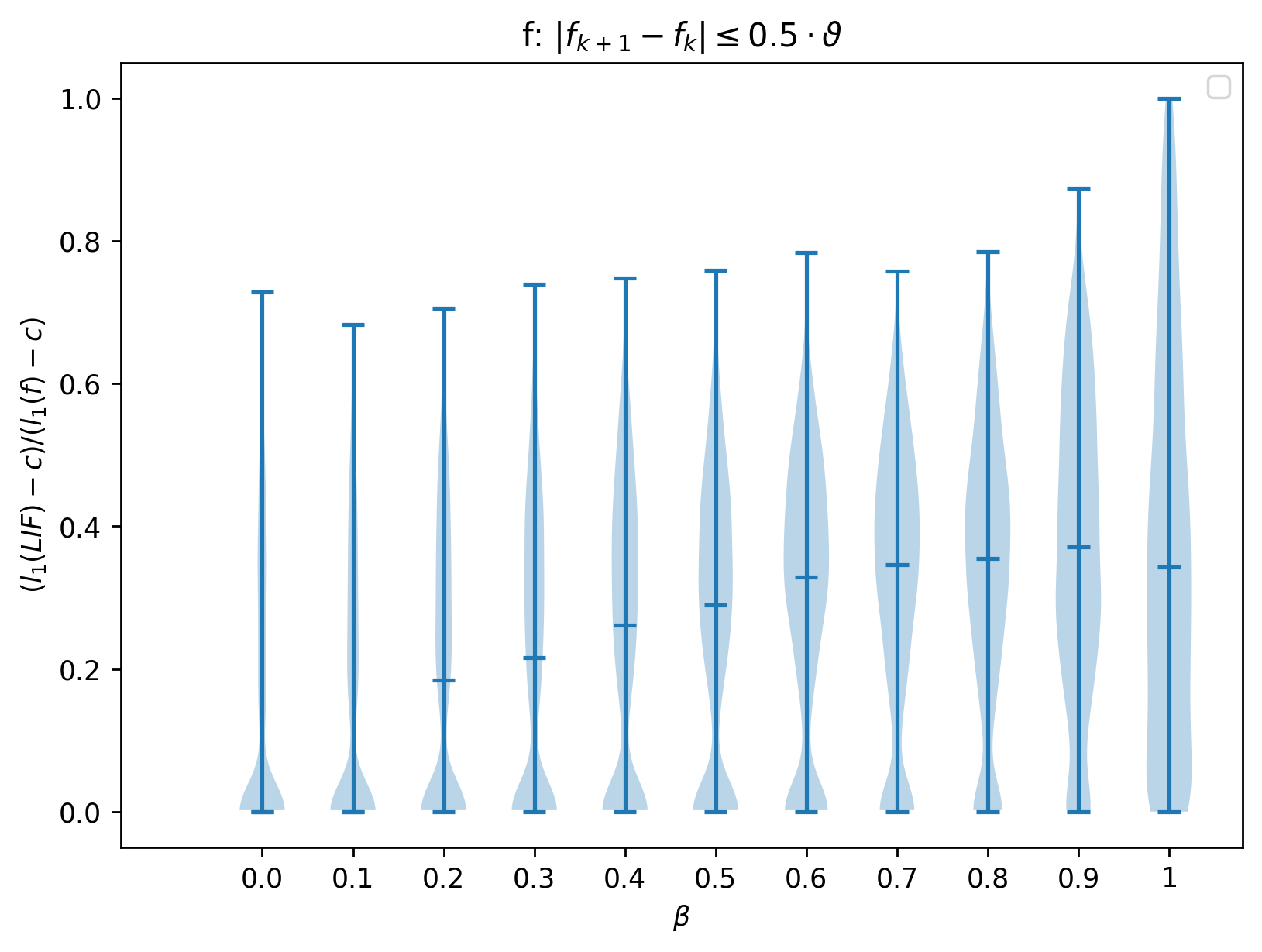}
		\caption{Distribution of $\lambda = (\|s\|_1 - c)/(\|f\|_1 - c)$ for $f$ satisfying the condition indicated in the title, 
		where $s = \mbox{LIF}_{\alpha, \vartheta}(f)$ and $c = \|B_{\alpha, \vartheta}(f)\|_1$.}
			\label{fig:ratio1}
\end{figure}
\begin{figure}[ht]
	\centering
	\includegraphics[width=0.45\linewidth]{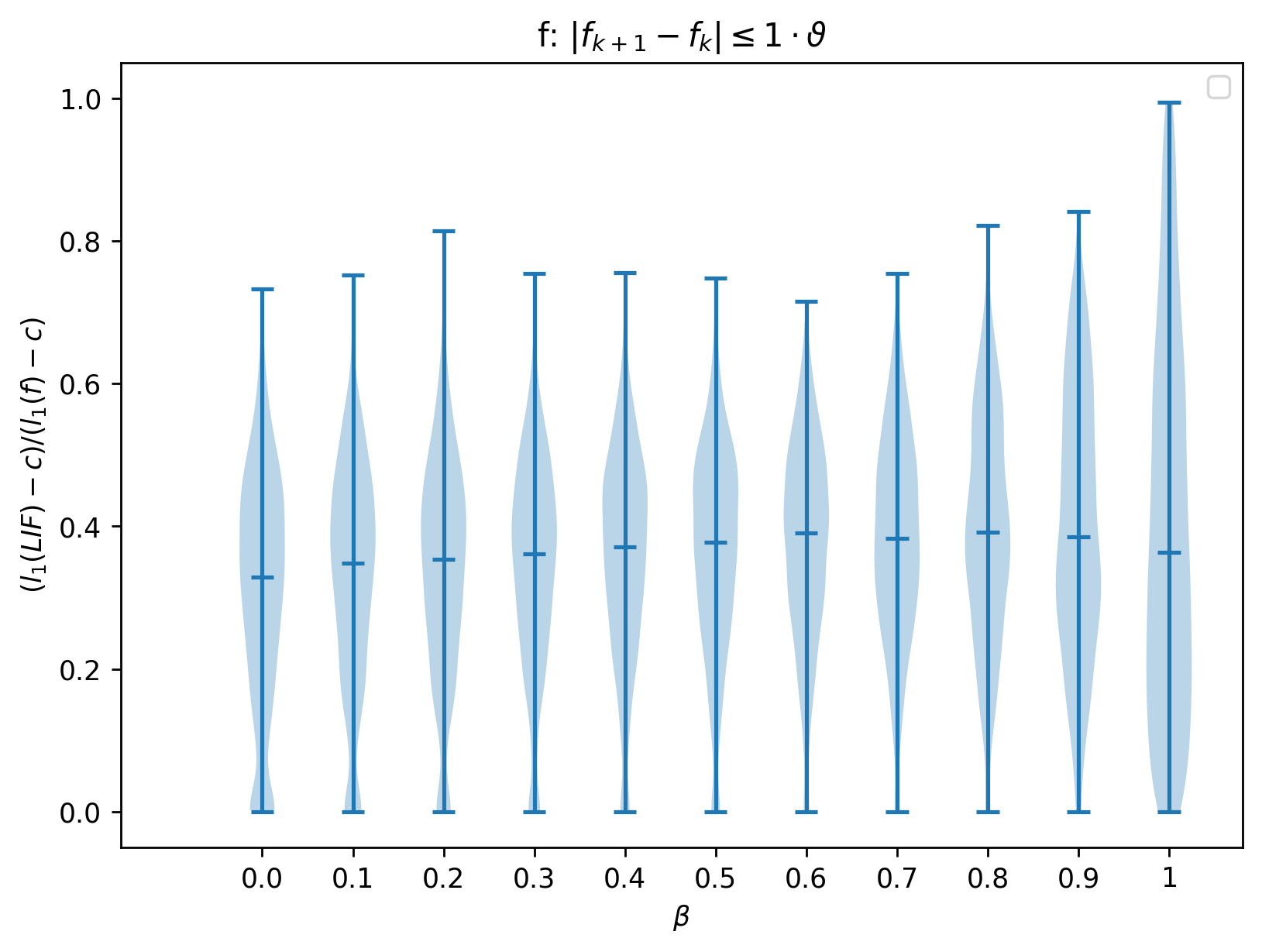}
	\includegraphics[width=0.45\linewidth]{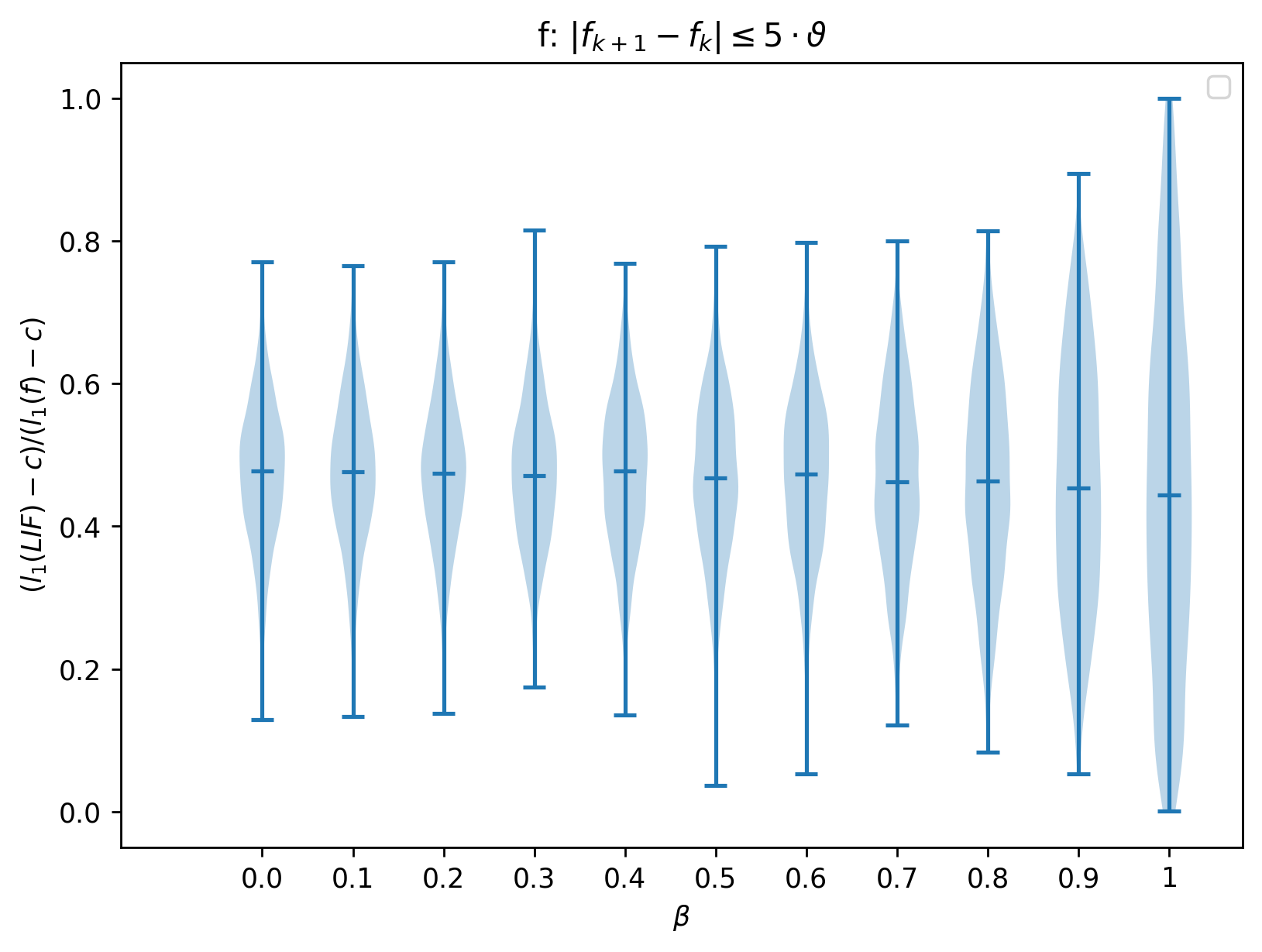}
		\caption{Like Fig.~\ref{fig:ratio1} for $f$ satisfying the condition indicated in the title. 
		}
			\label{fig:ratio2}
\end{figure}

Summing up, we find that LIF can be understood as quantization operator  
$Q_{\alpha, \vartheta}: \mathcal{F} \rightarrow \mathbb{S}_{\vartheta}$ that satisfies the following conditions:
\begin{itemize}
\item[(i)] {\bf Causality Condition.} $Q_{\alpha, \vartheta}$ is {\it causal}, i.e., 
							\begin{equation}
							\label{eq:Q1}
							\forall f, g\in \mathcal{F}: f|_{[0, T]} = g|_{[0, T]} \Rightarrow Q(f)|_{[0, T]} = Q(g)|_{[0, T]}.
							\end{equation}
\item[(ii)] {\bf Sparsity Condition.} 
							$Q_{\alpha, \vartheta}(f)$  is at least as sparse as $f$, i.e.,
							\begin{equation}
							\label{eq:Q2}
							\forall f\in \mathcal{F}: \|Q(f)\|_1 \leq \|f\|_1.
							\end{equation}
\item[(iii)] {\bf Alexiewicz Accuracy Condition.} $Q_{\alpha, \vartheta}$ satisfies the {\it Alexiewicz accuracy condition} 
							\begin{equation}
							\label{eq:Q3}
							\forall f\in \mathcal{F}: \|f - Q(f) \|_{A, \alpha} < \vartheta.
							\end{equation}
\end{itemize}

While the {\it Causality Condition}~(\ref{eq:Q1}) is a necessity for inline information processing
and the {\it Sparsity Condition}~(\ref{eq:Q2}) is necessary to justify the notion of sparsity,
the {\it Alexiewicz Accuracy Condition}~(\ref{eq:Q3}) is not that intuitive as it brings into play 
a novel notion of geometry that deviates from familiar settings in signal processing, see~\cite{MOSER2024128190, Moser2017Similarity, Moser12UnitBall}.
An alternative way to express~(\ref{eq:Q3}) is to postulate that $Q(f) \in \mathring{B}^A_{\alpha,\vartheta}(f) \cap \mathbb{S}_{\vartheta}$,
where $\mathring{B}^A_{\alpha,\vartheta}(f)  = \{g \in \mathcal{F}:\, \|f - g\|_{A, \alpha} < \vartheta\}$ is the open $\alpha$-Alexiewicz ball with radius $\vartheta>0$ 
centered at $f$.

The evaluations of Fig.~\ref{fig:ratio1}-\ref{fig:ratio2} raise the question 
how LIF behaves compared to other spike train quantization schemes as studied next. 

\section{Extremal Sparsity Property of Integrate-and-Fire}
\label{ss:IFQuant}
The question arises whether $\mbox{LIF}_{\alpha, \vartheta}$ can be directly represented as a transformed version of the 
standard quantization $q_{\vartheta}$. The idea is to check whether there holds $\mbox{LIF}_{\alpha, \vartheta}(f) = A_{\alpha}^{-1} \circ q_{\vartheta} \circ A_{\alpha} (f)$.
For simplicity, in a first step we restrict ourselves to discrete time by considering sequences $s = (s_1, s_2, \ldots) \in \mathbb{R}^{\mathbb{N}}$, which, for 
sake of convenience, we extend to $s = (s_0, s_1, s_2, \ldots) \in \mathbb{R}^{\mathbb{N}_0}$ with $s_0=0$ by default.
Then the linear transformation (\ref{eq:Acont}) becomes an ordinary sum
\begin{equation}
\label{eq:A}
\hat{s}:= A_{\alpha}(s):= (\sum_{j=0}^k e^{-\alpha\, (k-j)} s_j)_k.
\end{equation}
Note that its inverse is given by $s = A_{\alpha}^{-1}(\hat{s}):= (\hat{s}_k - e^{-\alpha}  \hat{s}_{k-1})_k$, provided that the involved sums are bounded.
For sake of simplified notation, we write $A = A_{0}$, resp. $A^{-1} = A_{0}^{-1}$ in the special case $\alpha = 0$.

Now we consider the case of zero leakage $\alpha = 0$, for which we state the following decomposition theorem for integrate-and-fire with threshold ${\vartheta}>0$. 
\begin{theorem}[Quantization Representation of Integrate-and-Fire]
\label{th:IFDecomposition}
For any $s = (s_0, s_1, \ldots ) \in \mathbb{R}^{\mathbb{N}_0}$ satisfying $\sup_n |\sum_{i=0}^n s_i | < \infty$,
with $A(s):= (\sum_{j=0}^k s_j)_k$ and $A_{\alpha}^{-1}(\hat{s}):= (\hat{s}_k - \hat{s}_{k-1})_k$
integrate-and-fire (IF) can be represented as a transformed quantization according to
\begin{equation}
\label{eq:decIF}
\mbox{IF}_{\vartheta}(s) = A^{-1}\circ q_{\vartheta} \circ A (s),
\end{equation}
i.e., $A(\mbox{IF}_{\vartheta}(s)) = \sum_{i=0}^n \mbox{IF}_{\vartheta}(s)|_{i} = q_{\vartheta}(\sum_{i=0}^n s_i)$,
where $q_{\vartheta}(s_0, s_1, \ldots ) = (q_{\vartheta}(s_0), q_{\vartheta}(s_1), \ldots )$ is applied element-wise and
$(s_0, s_1, \ldots )|_{i}:= s_i$.
\end{theorem}

\begin{proof}
Let denote the membrane's potential at index $n$ by $P_n$.
It is built recursively by adding the next function value to the actual potential except the sum exceeds the threshold then
as many times the threshold is subtracted until the resulting difference become below threshold.
This also explains the spike triggering process and the spike amplitude fired at $i=n+1$, which is given by 
\begin{equation}
\label{eq:recIF}
\mbox{IF}(s)|_{n+1} = q(P_n +  s_{n+1}).
\end{equation}
So, for $i=0$ we get $P_0 = s_0 - q(s_0)$, and, for $i=n+1$ we get 
\begin{equation}
\label{eq:recPot}
P_{n+1} = P_{n} + s_{n+1} - q(P_{n} + s_{n+1})
\end{equation}
By induction, we get for all $i$
\begin{equation}
\label{eq:pot}
P_i = s_0 + \ldots + s_i - q(s_0 + \ldots + s_i), 
\end{equation}
which follows from the fact that $q(a \pm q(b)) = q(a) \pm q(b)$
and the observation that the induction assumption for $i=n$ applied on (\ref{eq:recPot}) yields
\begin{eqnarray}
P_{n+1} & = & s_0 + \ldots + s_n - q(s_0 + \ldots + s_n) + s_{n+1} - \nonumber \\ 
        &  &  \,\,\,q(s_0 + \ldots + s_n - q(s_0 + \ldots + s_n) + s_{n+1}) \nonumber \\       
				& = & s_0 + \ldots + s_n + s_{n+1}  - q(s_0 + \ldots + s_n + s_{n+1} ). \nonumber
\end{eqnarray}
Now, using $(\ref{eq:pot})$ to reformulate (\ref{eq:recIF}) we get 
\begin{eqnarray}
\label{eq:LIrec}
\mbox{IF}(s)|_{n+1} & = & q(P_n +  s_{n+1}) \nonumber \\
									 & = & q(s_0 + \ldots + s_n - q(s_0 + \ldots + s_n) +  s_{n+1}) \nonumber \\
									 & = & q(s_0 + \ldots + s_{n+1}) - q(s_0 + \ldots + s_n). 
\end{eqnarray}
By using~(\ref{eq:LIrec}), summing up $\mbox{IF}(s)|_{0}$ up to $\mbox{IF}(s)|_{n+1}$ results in the telescope sum
$q(s_0+s_1) - q(s_0) + q(s_0+s_1+s_2) - q(s_0+s_1) + \ldots  + q(s_0+ \ldots +s_{n+1}) - q(s_0 + \ldots s_n)$, 
which due to $q(s_0)=0$ becomes $q(s_0+ \ldots +s_{n+1}) - q(s_0) = q(s_0+ \ldots +s_{n+1})$, proving (\ref{eq:decIF}).
\end{proof}

In a similar way, this quantization decomposition property can be generalized to functions $f:[0, \infty) \rightarrow \mathbb{R}$ by means of the linear transformation
(\ref{eq:Acont}), yielding $\mbox{IF}_{\vartheta}(f)= A^{-1}\circ q_{\vartheta} \circ A (f)$. Equ. (\ref{eq:decIF}) is useful to get a geometric understanding of the quantization operator of IF, resp. later on LIF, which becomes manifest by considering
\begin{equation}
\label{eq:quantizationForm}
\|q_{\vartheta}\circ A(f) - A(f)\|_{\infty} = \|A\circ \mbox{IF}_{\vartheta}(f) - A(f)\|_{\infty} = \|\mbox{IF}_{\vartheta}(f) - f\|_{A} < \vartheta,
\end{equation}
where $\|f\|_A := \sup_T |\int_0^T f(t)dt|$ is the Alexiewicz norm~\cite{Alexiewicz1948,MOSER2024128190}. 
With Equ. (\ref{eq:quantizationForm}) we regain the quantization error bound formula, $\|\mbox{IF}_{\vartheta}(f) - f\|_{A} < \vartheta$, 
for IF in a direct way. It also shows how the Alexiewicz norm comes into play, 
which provides a different mathematical motivation than the topological argument given in~\cite{MOSER2024128190}.

\subsection{Extremal Sparsity Property of Integrate-and-Fire (IF) and Send-on-Delta (SOD)}
\label{ss:zero}
This  brings us now to the point to state the extremal sparseness property for IF in terms of the $l_1$-norm, 
$\|(s_0, s_1, \ldots)\|_1 = \sum_k |s_k|$. For the proof see the Appendix C.
\begin{theorem}[Extremal Sparsity Property of Integrate-and-Fire]
\label{th:sparsenessIF}
For any function $f:[0, T] \rightarrow \mathbb{R}$  that is integrable and almost everywhere bounded with locally finite many superimposed Dirac impulses, 
integrate-and-fire with threshold $\vartheta>0$ satisfies the maximal sparsity property
\begin{equation}
\label{eq:sparsenessProp}
\|IF_{\vartheta}(f)\|_1 = \min\{\|s\|_1:  s \in \mathring{B}^A_{\vartheta}(f) \cap \mathbb{S}_{\vartheta}  \},
\end{equation}
where $IF_{\vartheta}(f):= IF_{\alpha, \vartheta}(f)$ for zero leakage $\alpha = 0$,  $\mathring{B}^A_{\vartheta}(f):= \{g \in \mathcal{F}: \|g - f\|_A < \vartheta\} $ 
represents the open Alexiewicz unit ball with radius $\vartheta$ centered at $f$ and 
$\mathbb{S}_{\vartheta}$ is the set of all spike trains $s$ with $s(t) = \sum_k s_k \, \delta(t - t_k)$, $s_k \in \vartheta\,\mathbb{Z}$.
\end{theorem}

As a corollary of Theorem~\ref{th:sparsenessIF} we immediately obtain the maximal sparsity property of SOD when applied on
Lipschitz continuous functions. First, recall that a real $f:[a, b] \rightarrow \mathbb{R}$ is said to be Lipschitz continuous if there is a constant $K>0$ such that for all $x \neq y \in [a,b]$ all fractions $\frac{f(x) - f(y)}{x - y}$ are bounded by $K$. 
The condition of Lipschitz continuity ensures that the sampling process based on SOD is well defined, 
in particular there is a minimal positive length for the time intervals between two subsequent sampling points.
Due to the theorem of Rademacher~\cite{Rademacher1919} any Lipschitz continuous function $f$ is differential almost everywhere on $[a,b]$, i.e.,
$f'$ does exist except on a set of Lebesgue measure zero. Further, the second fundamental theorem of calculus is valid for its derivation $f'$, i.e. $f'$ is integrable and satisfies $\int_a^b f' dt = f(b) - f(a)$. See also~\cite{Zurcher2007} and the proof of Theorem 6.15 in ~\cite{Heinonen2001}.
Now, applying Theorem~\ref{th:sparsenessIF} on $f'$ yields the desired 

\begin{corollary}[Extremal Sparsity Property of SOD]
\label{cor:sparsenessSOD}
For any Lipschitz continuous function $f:[0, T] \rightarrow \mathbb{R}$, SOD-based threshold-based sampling with threshold $\vartheta>0$ satisfies the maximal sparsity property
\begin{equation}
\label{eq:sparsenessProp1}
\|f(t_k) - f(t_{k+1})\|_1 = \min\{\|s\|_1:  s \in \mathring{B}^{\|.\|_{\infty}}_{\vartheta}(\mbox{Off}[f]) 
\cap \mathbb{S}_{\vartheta}  \},
\end{equation}
where $\mbox{Off}[f] := f - f(0)$ and $\mathring{B}^{\|.\|_{\infty}}_{\vartheta}(f):= \{g \in \mathcal{F}: \|g - f\|_{\infty} < \vartheta\}$ and 
$\mathbb{S}_{\vartheta}$ is defined as in Theorem~\ref{th:sparsenessIF}.
\end{corollary}

\section{Sparsity in the General Case}
\label{s:GeneralCase}
The Alexiewicz Accuracy Condition~(\ref{eq:Q3}) can be decomposed into a recursive decision problem. 
Note that the proof for $\vartheta = 1$ also generalizes to arbitrary thresholds
$\vartheta>0$.
To select an admissible spike train $s = (s_1, s_2, \ldots)$ we have to check
\begin{eqnarray}
\label{eq:AlexAlternative}
-1   < f_1 - s_1  <  1, \nonumber \\
-1   < f_{k+1} - s_{k+1} + \beta (z_k - s_k) <  1, \nonumber
\end{eqnarray}
hence,
\begin{eqnarray}
\label{eq:AlexAlternative1}
f_1 -1  < s_1 <  f_1 + 1, \nonumber \\
f_{k+1} + \beta (z_k - s_k)  -1   <  s_{k+1} <  f_{k+1} + \beta (z_k - s_k)  +1, \nonumber
\end{eqnarray}
which leaves two options 
\begin{eqnarray}
\label{eq:AlexAlternative2}
s_{k+1} \in \{q_{k+1}^{(1)}, q_{k+1}^{(2)}\}
\end{eqnarray}
in case of $f_{k+1} + \beta (z_k - s_k) \not \in \mathbb{Z}$, 
where $q_{k+1}^{(1)} = q(f_{k+1} + \beta (z_k - s_k))$ and 
$q_{k+1}^{(2)} = q(f_{k+1} + \beta (z_k - s_k)) + \mbox{sgn}(f_{k+1} + \beta (z_k - s_k))\}$.
In case of $f_{k+1} + \beta (z_k - s_k) \in \mathbb{Z}$ the choice is uniquely determined by 
$s_{k+1} := q_{k+1}^{(1)}$. 

This way, (\ref{eq:AlexAlternative2}) can be utilized to enumerate all grid points representing admissible spike trains in the open Alexiewicz ball to examine experimental evaluations, see Fig.~\ref{fig:BoxPlot1}-\ref{fig:BoxPlot3a}.
As a next step we look at an example that illustrates that the extremal sparsity property of IF cannot hold in the general case.

For this we utilize (\ref{eq:AlexAlternative2})  to generate an alternative to LIF
by choosing $s_1 := q_1^2$ and $s_k := q_k^1$ for $k>1$. Consider a step function 
$h(t):=\sum_k (-1)^{k+1} a_{(\alpha)} 1_{[k, k+1)}(1)$, where $\int_{0}^1 a_{(\alpha)}\, e^{-\alpha \, (1-t) }dt = 1$, i.e.,
$a_{\alpha}:= \alpha/(1- e^{-\alpha})$. In this case $\mbox{LIF}_{\alpha, 1}(h)$ yields a spike train of alternating signs with unit 
time intervals in-between. In this case the spike train $s^*$ given by $q^2_0, q^1_1, q^1_2, \ldots$ reduces to a single spike at $t=0$,
i.e., $s^* = 1\, \delta(t)$. See Fig.~\ref{fig:counterExample} for a discrete variant of this example.
For this example we get $\|\mbox{LIF}_{\alpha, 1}(h)|_{[0, n]}\|_1 = n$ while $\|s^*\|_1 = 1$. 
This means that there are some constellations for which leaky integrate-and-fire is not at all sparse.
However, applying some noise relatives this counter-example, see Fig.~\ref{fig:counterExampleNoise}, indicating that under reasonable conditions LIF can perform in a sparse way. As the evaluations of Fig.~\ref{fig:BoxPlot1}-\ref{fig:BoxPlot3a} demonstrate, the more smooth $f$ the higher the probability that for any spike-train quantization operator $Q_{\alpha, \vartheta}$ satisfying the conditions (\ref{eq:Q1}), (\ref{eq:Q2}) and (\ref{eq:Q3}) the difference $\| Q_{\alpha, \vartheta}(f)\|_1 - \| \mbox{LIF}_{\alpha, \vartheta}(f)\|_1$ becomes non-negative. Note that the mass above the red line in the Figures~\ref{fig:BoxPlot1}-\ref{fig:BoxPlot3a} represents the probability  
that LIF yields extremal sparsity, which in all our experiments is observed to be above $0.99$. 
 
\begin{figure}[ht]
	\centering
	\includegraphics[width=0.45\linewidth]{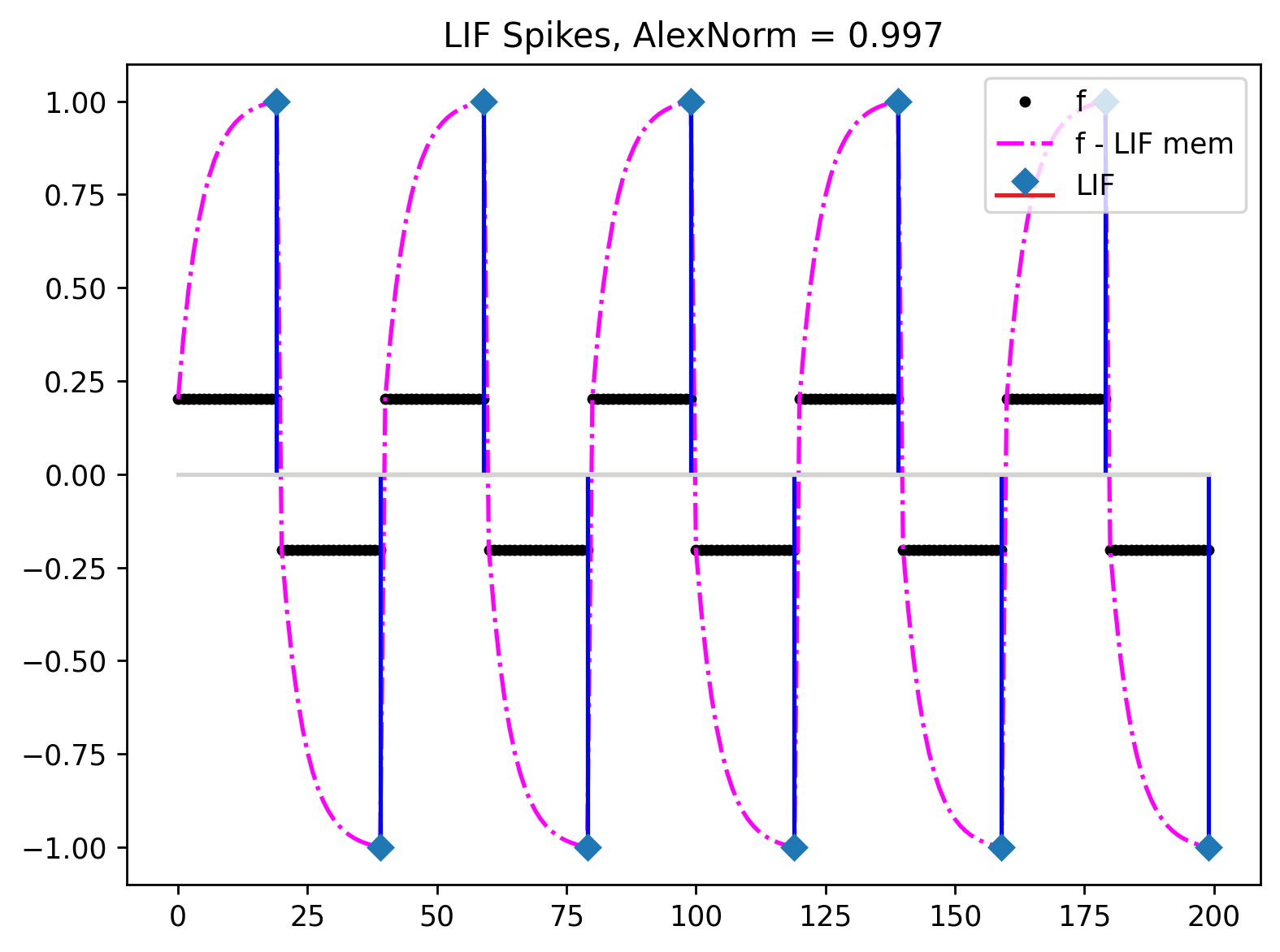}
	\includegraphics[width=0.45\linewidth]{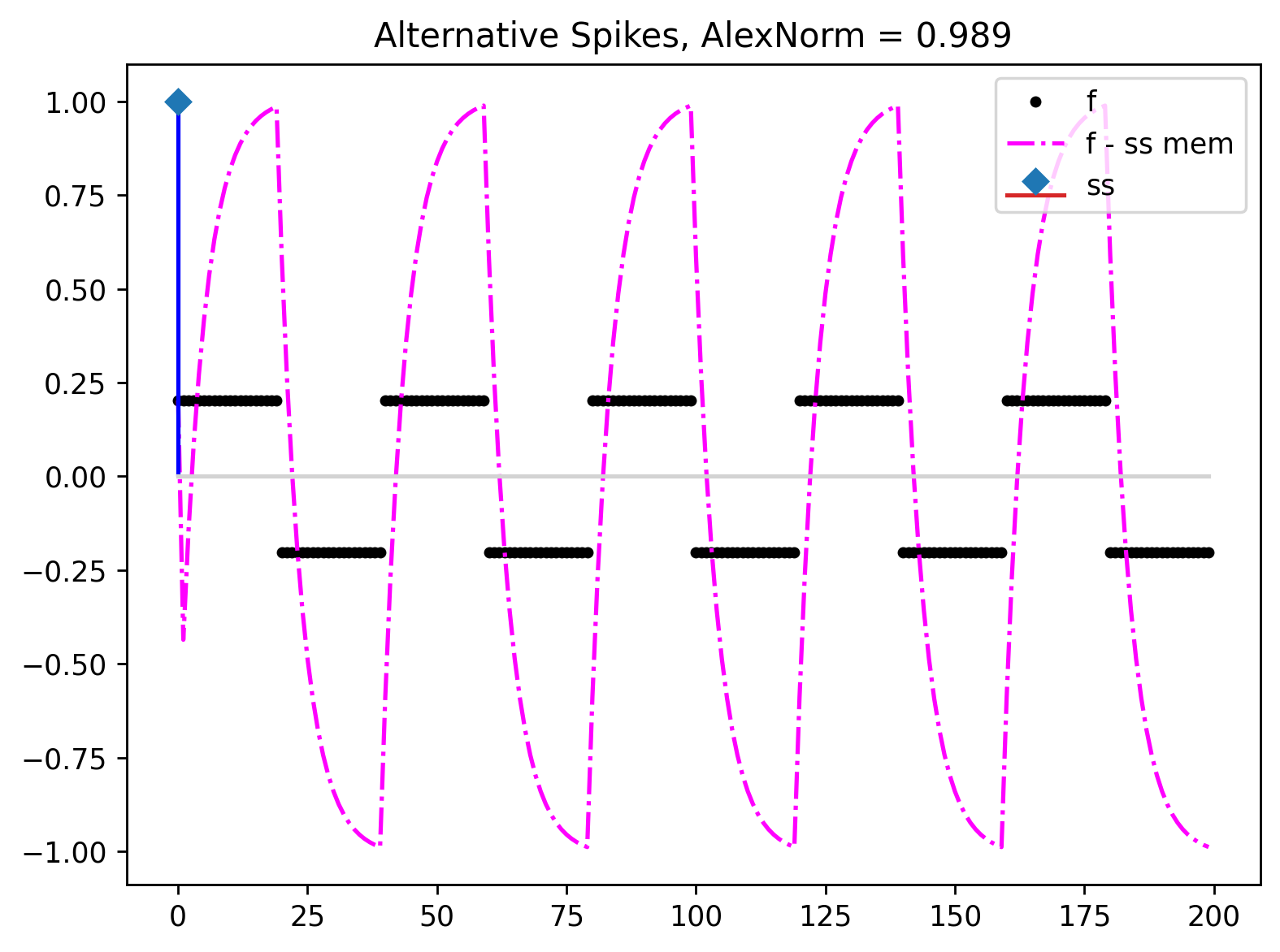}
		\caption{Example that for which LIF generates a spike train of alternating signs, here with $\beta = 0.8$.
		As a consequence, the alternative spike train quantization operator $Q(f)$ given by $q^2_0, q^1_1, q^1_2, \ldots$ in the scheme 
		(\ref	{eq:AlexAlternative2}) consists only of a single spike, that suffices to keep the signal below threshold.
		}
			\label{fig:counterExample}
\end{figure}

\begin{figure}[ht]
	\centering
	\includegraphics[width=0.45\linewidth]{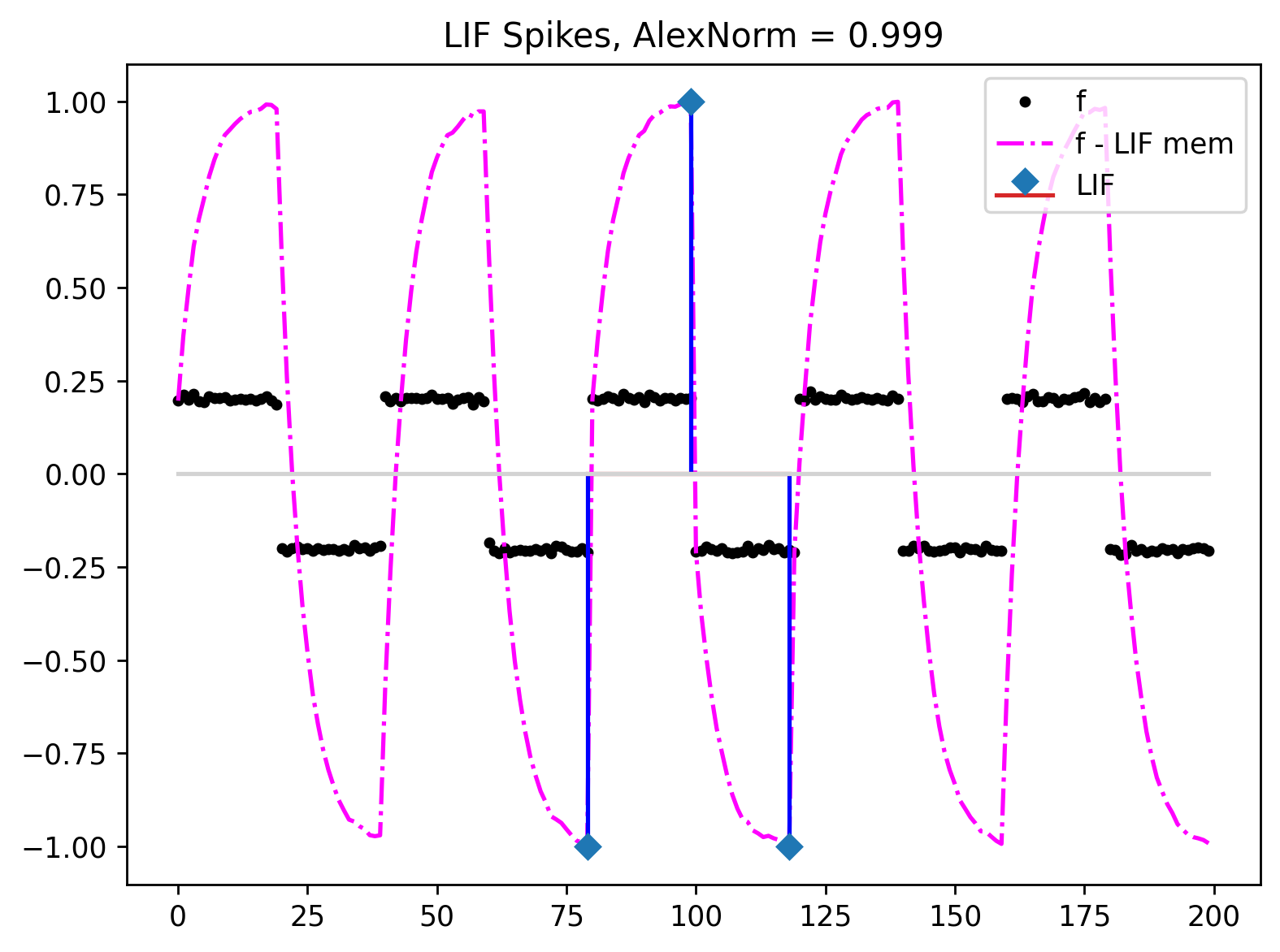}
	\includegraphics[width=0.45\linewidth]{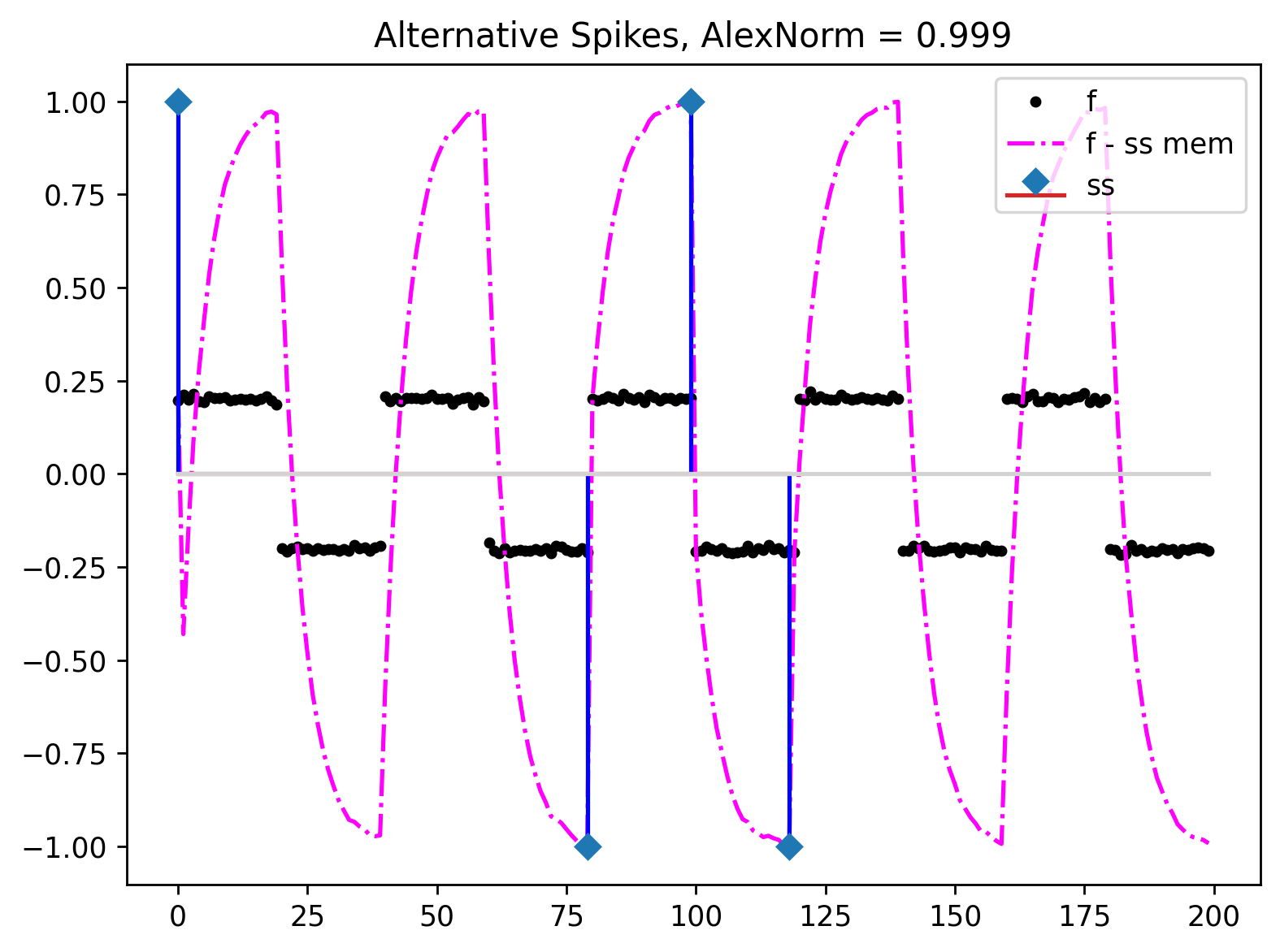}
		\caption{Same example as in Fig.~\ref{fig:counterExample} but with additive Gaussian noise $\mathcal{N}(0, a/33)$.
		This noise brings LIF back into a sparsity regime.
		}
			\label{fig:counterExampleNoise}
\end{figure}

\begin{figure}[ht]
	\centering	
	\includegraphics[width=0.45\linewidth]{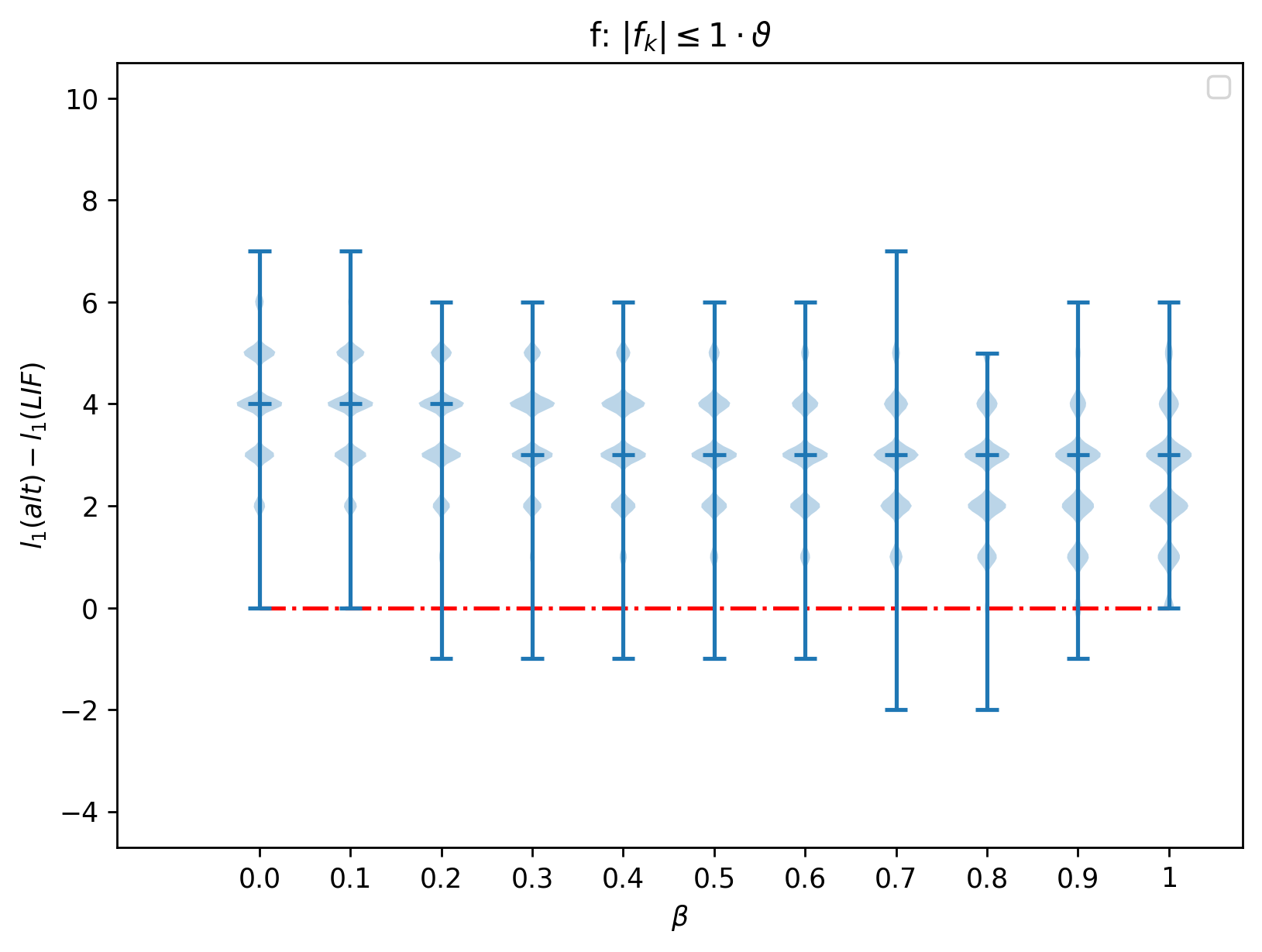}	
	\includegraphics[width=0.45\linewidth]{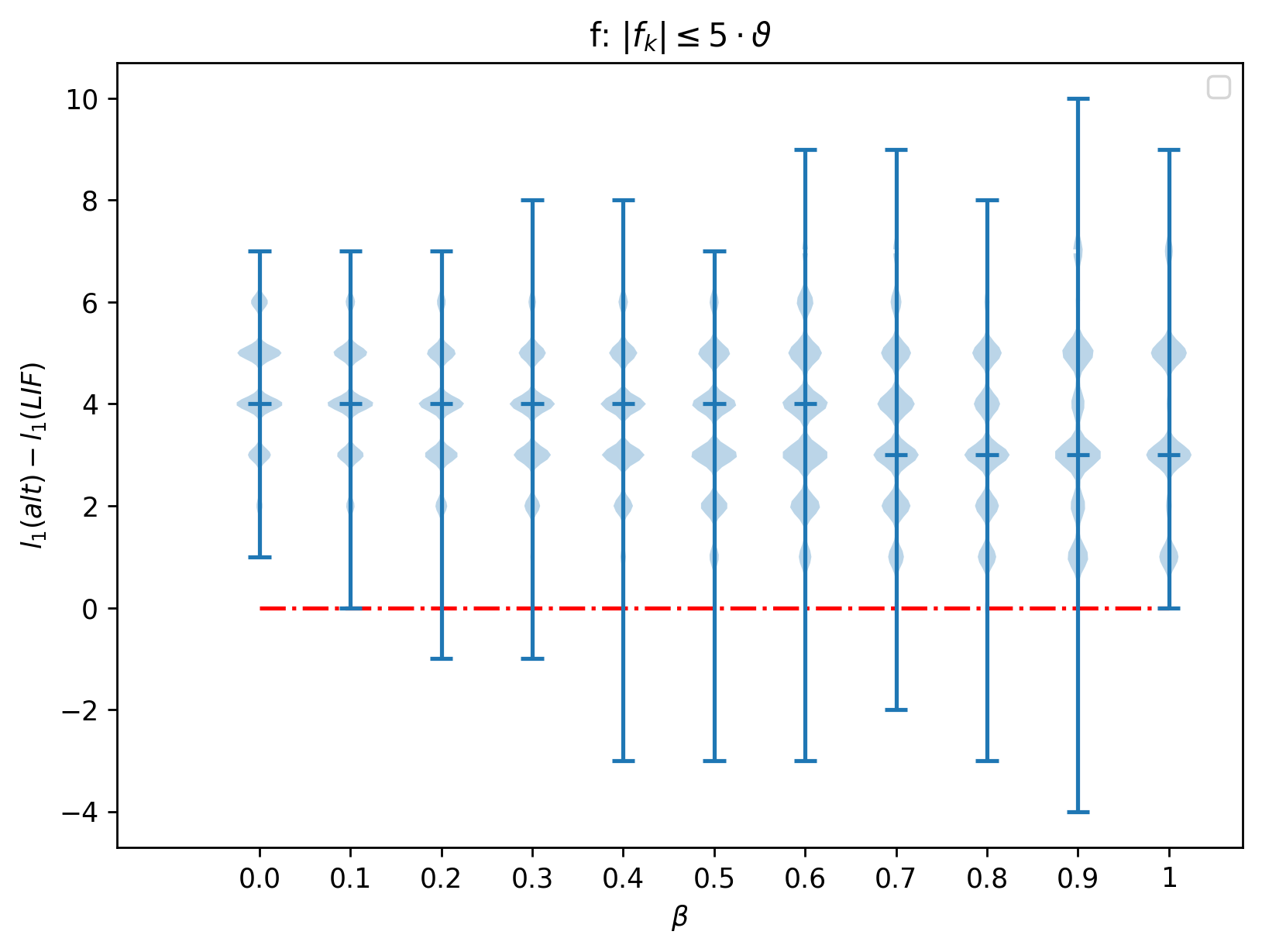}	
		\caption{
		Experimental test of the extremal sparsity property of LIF by exhaustive search of 
		spike train quantization operators $Q_{\beta, \vartheta}$ satisfying (\ref{eq:Q1}), (\ref{eq:Q2}) and (\ref{eq:Q3}).   
		The figures show violin-plots of $\|Q_{\beta, \vartheta}(f)\|_1 -\|\mbox{LIF}_{\beta, \vartheta}(f)\|_1$ 
		for random functions $f$ satisfying $|f_{i}|\leq K\,\vartheta$ ($\vartheta = 1$, $i\leq 10$)
		for different $\beta$. Left: $K= 1$, right: $K=5$. }
			\label{fig:BoxPlot1}
\end{figure}
\begin{figure}[ht]
	\centering
	\includegraphics[width=0.45\linewidth]{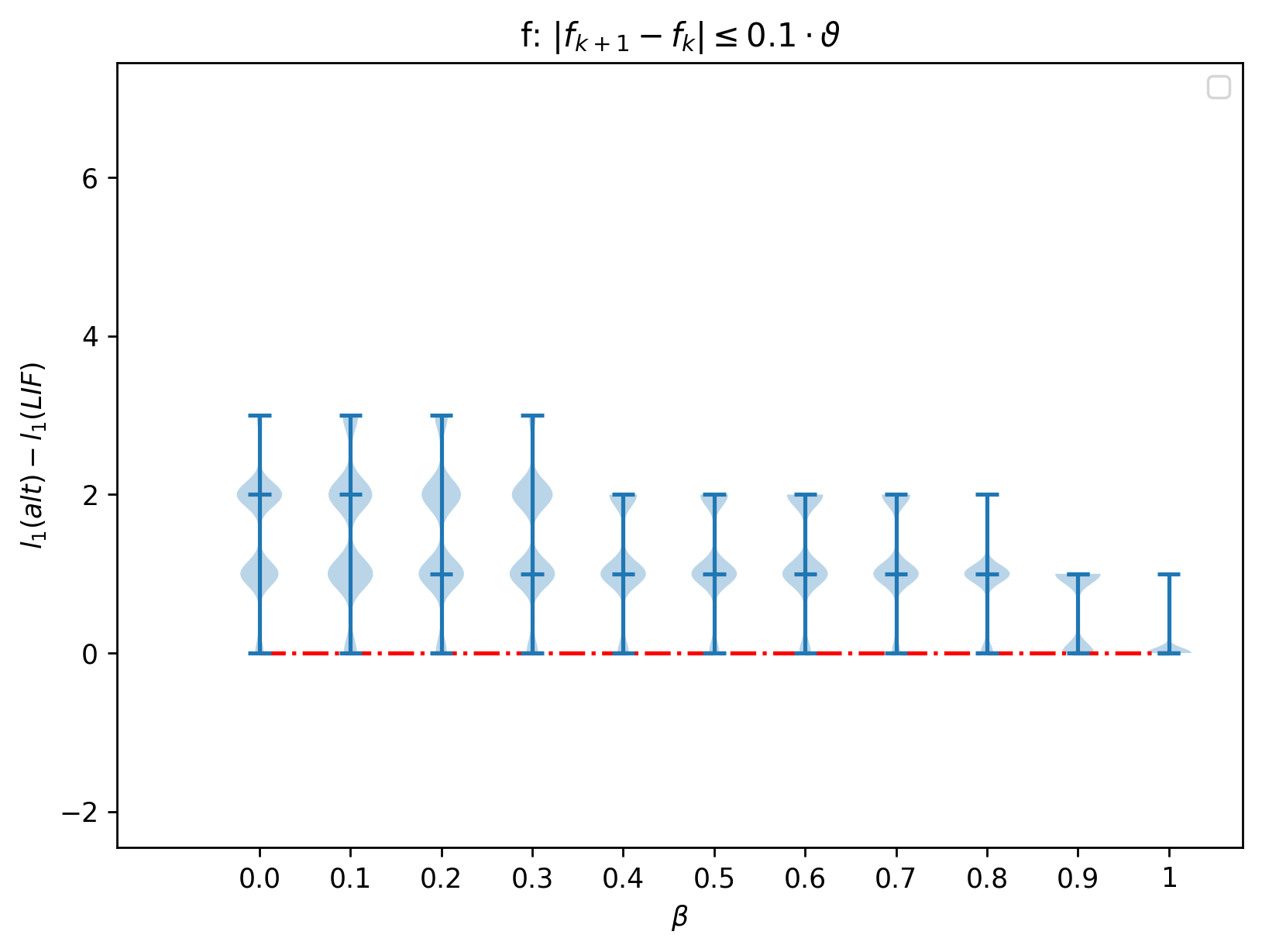}	
	\includegraphics[width=0.45\linewidth]{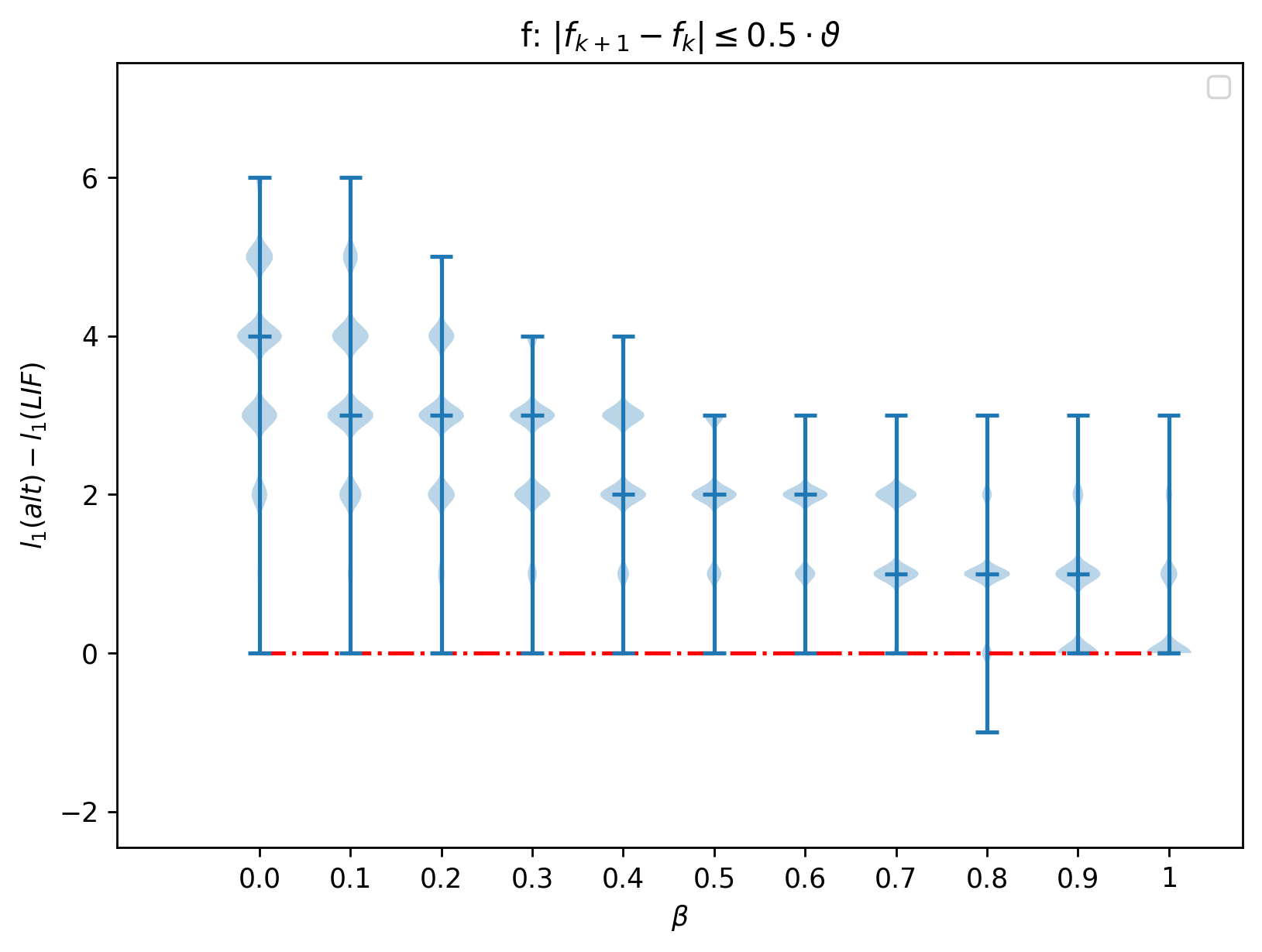}	
		\caption{
		Like Fig.~\ref{fig:BoxPlot2} but for $f$ satisfying $|f_{i+1}-f_i|\leq K\,\vartheta$ ($\vartheta = 1$, $i\leq 10$)
		for different $\beta$. Left: $K=0.1$, right: $K=0.5$}
			\label{fig:BoxPlot2}
\end{figure}
\begin{figure}[ht]
	\centering
	\includegraphics[width=0.45\linewidth]{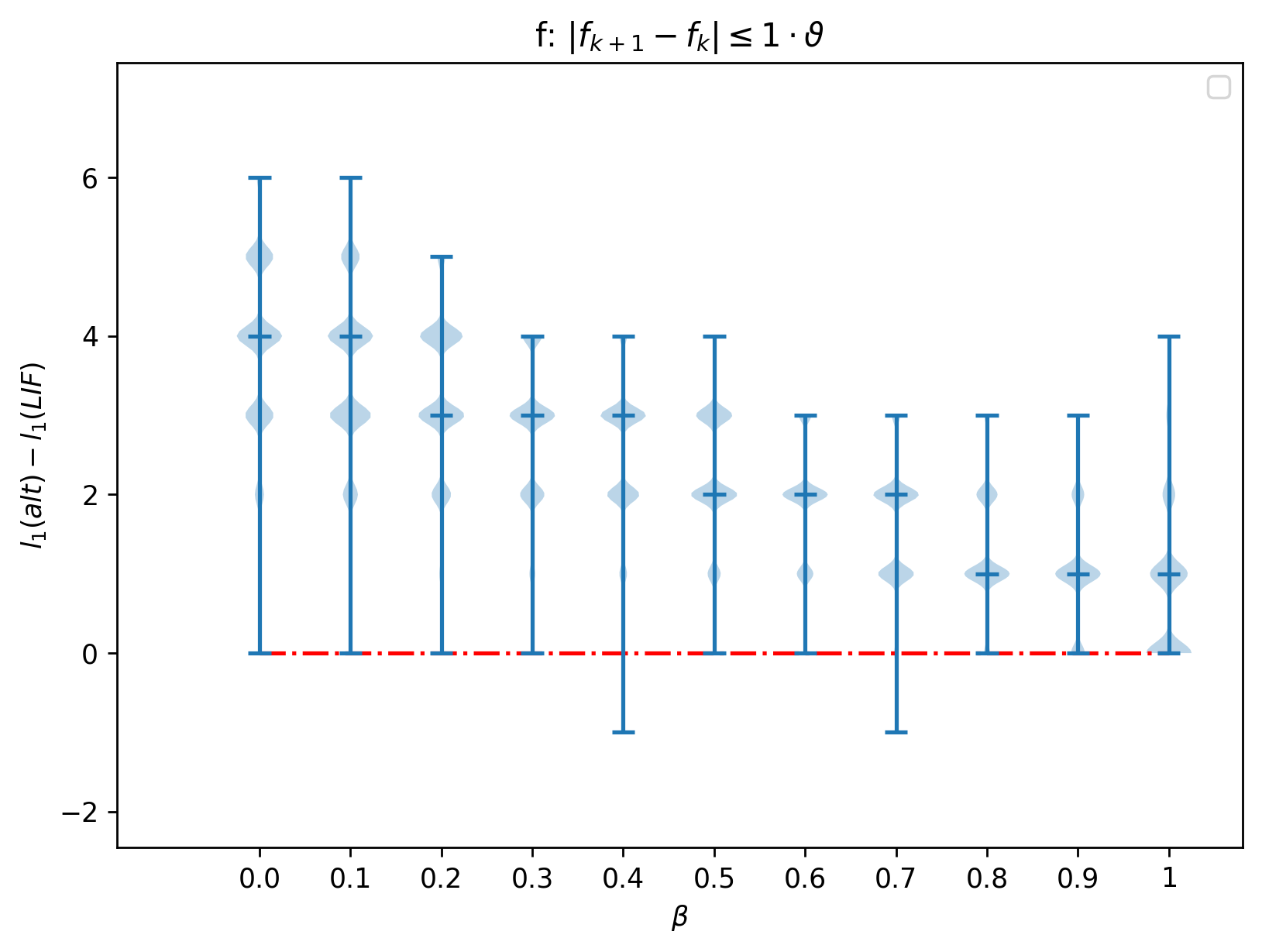}	
	\includegraphics[width=0.45\linewidth]{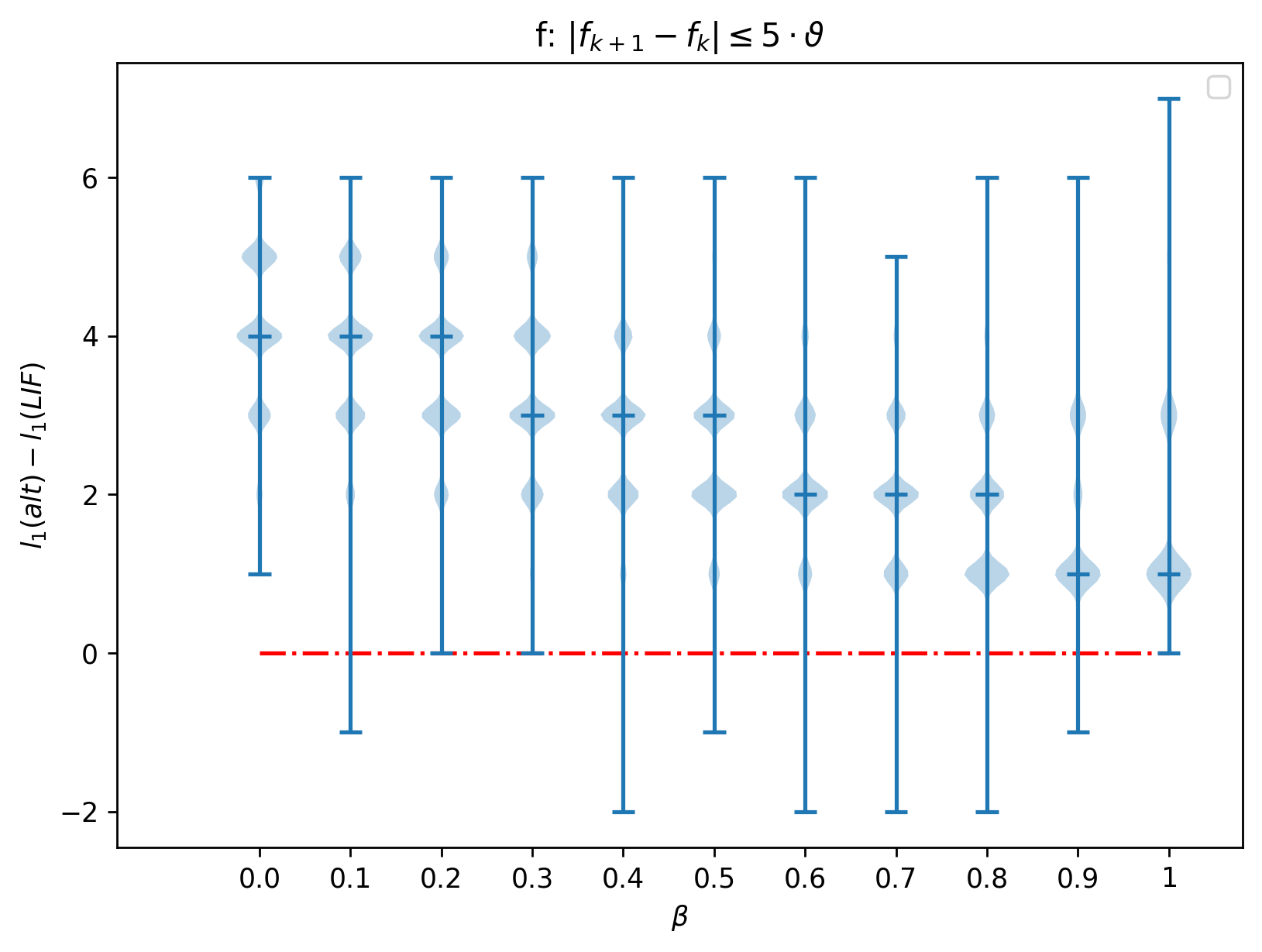}	
		\caption{
		Continuation of Fig.~\ref{fig:BoxPlot2}. Left: $K= 1$, right: $K=5$.}
			\label{fig:BoxPlot2a}
\end{figure}
\begin{figure}[ht]
	\centering	
	\includegraphics[width=0.45\linewidth]{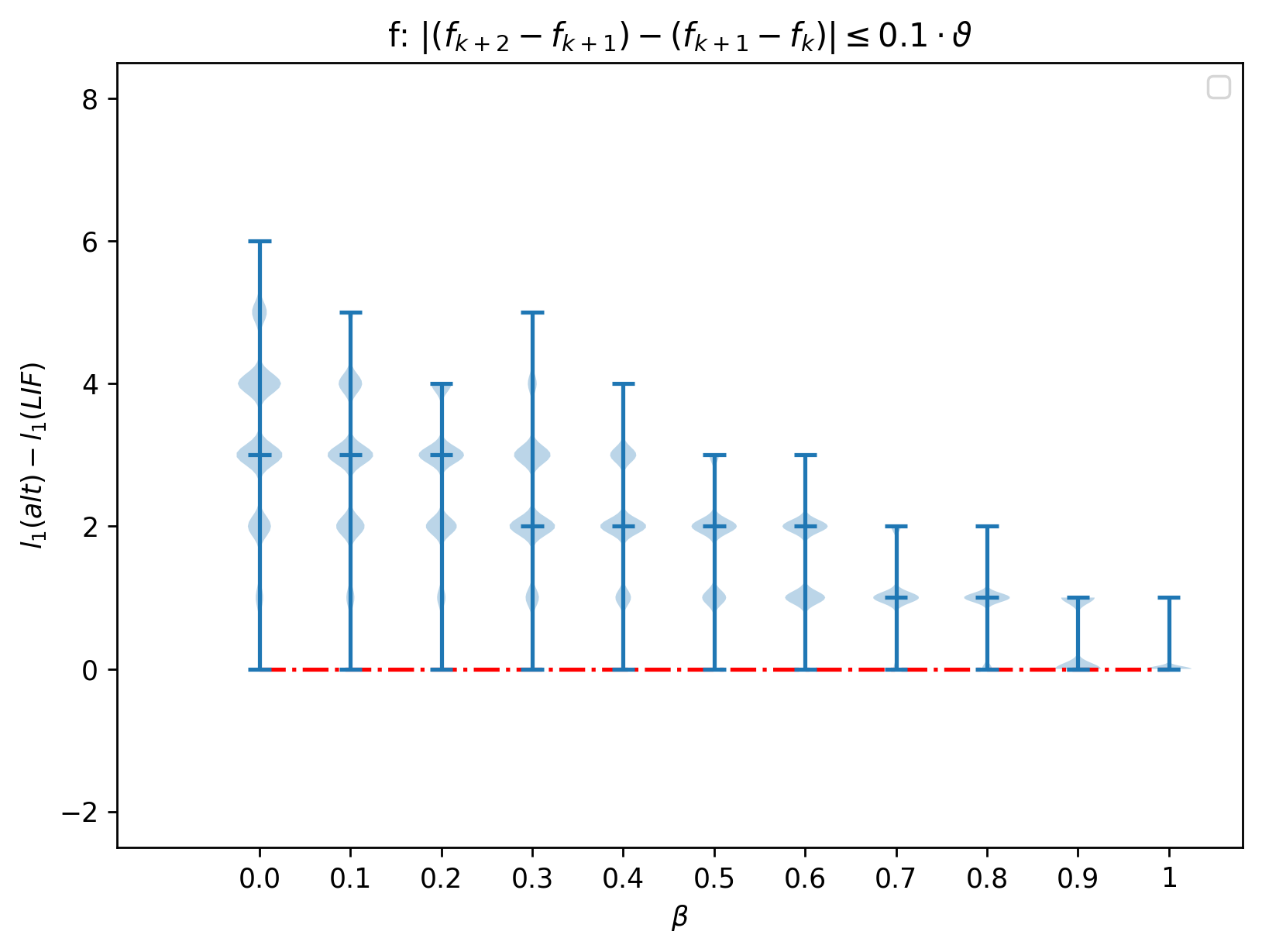}	
	\includegraphics[width=0.45\linewidth]{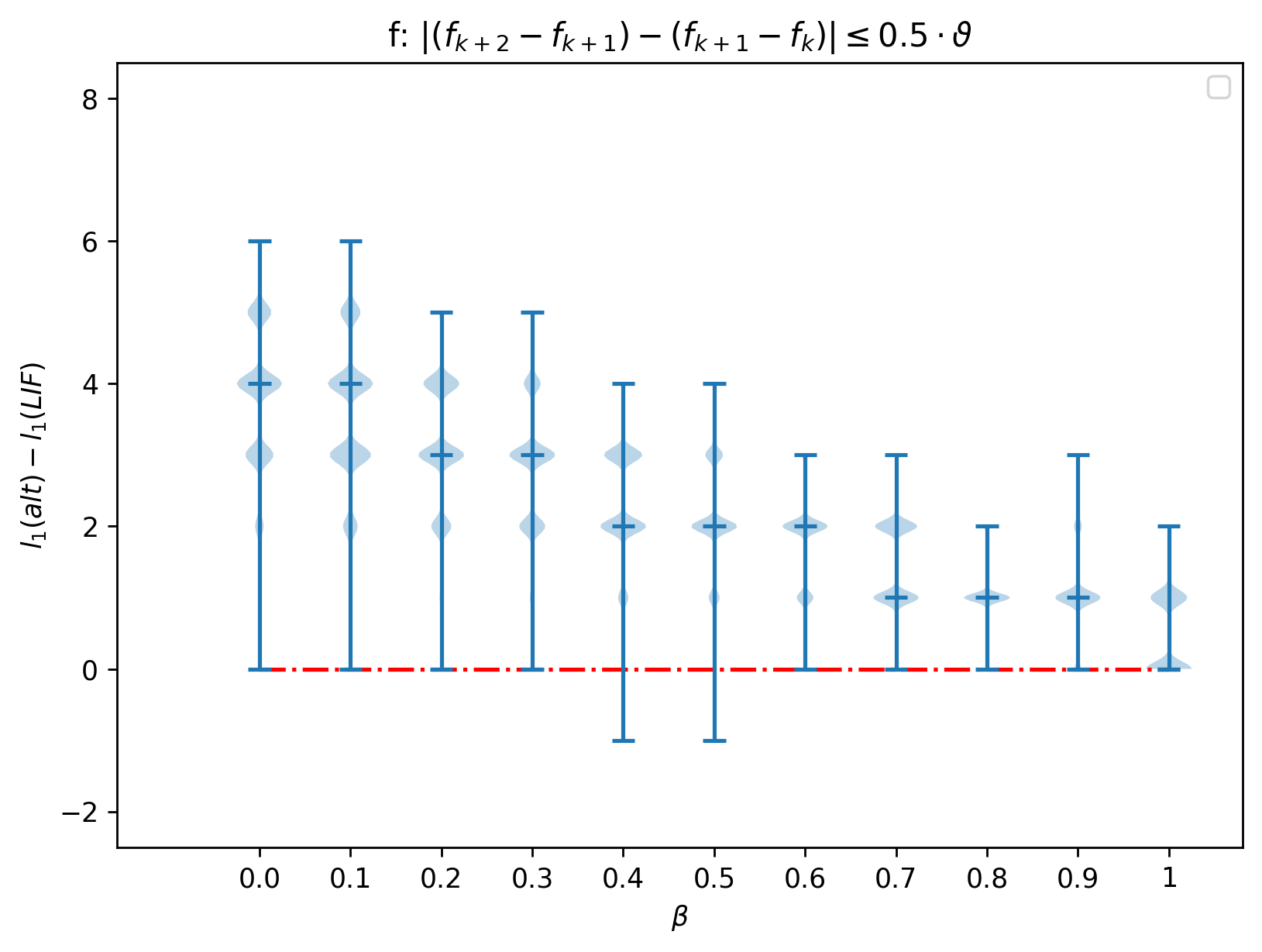}	
		\caption{Like Fig.~\ref{fig:BoxPlot1} and Fig.~\ref{fig:BoxPlot1} but for $f$ satisfying a bound of differences of second order, i.e.,
		$|(f_{k+2}- 2\, f_{k+1} + f_k|\leq K\,\vartheta$ ($\vartheta = 1$, $k\leq 10$) for different $\beta$. 
		Left: $K=0.1$, right: $K=0.5$.
		}
			\label{fig:BoxPlot3}
\end{figure}
\begin{figure}[ht]
	\centering	
	\includegraphics[width=0.45\linewidth]{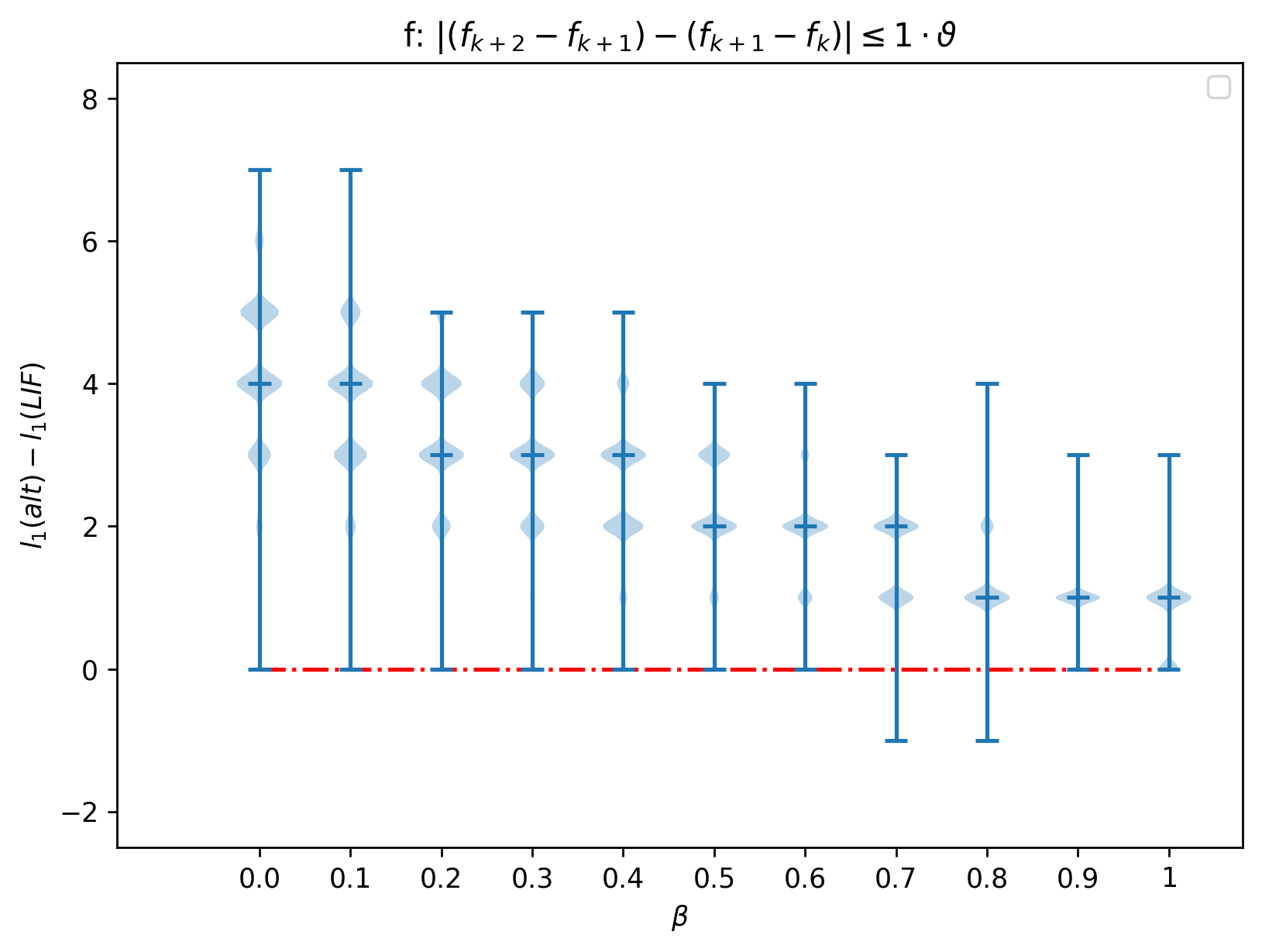}	
	\includegraphics[width=0.45\linewidth]{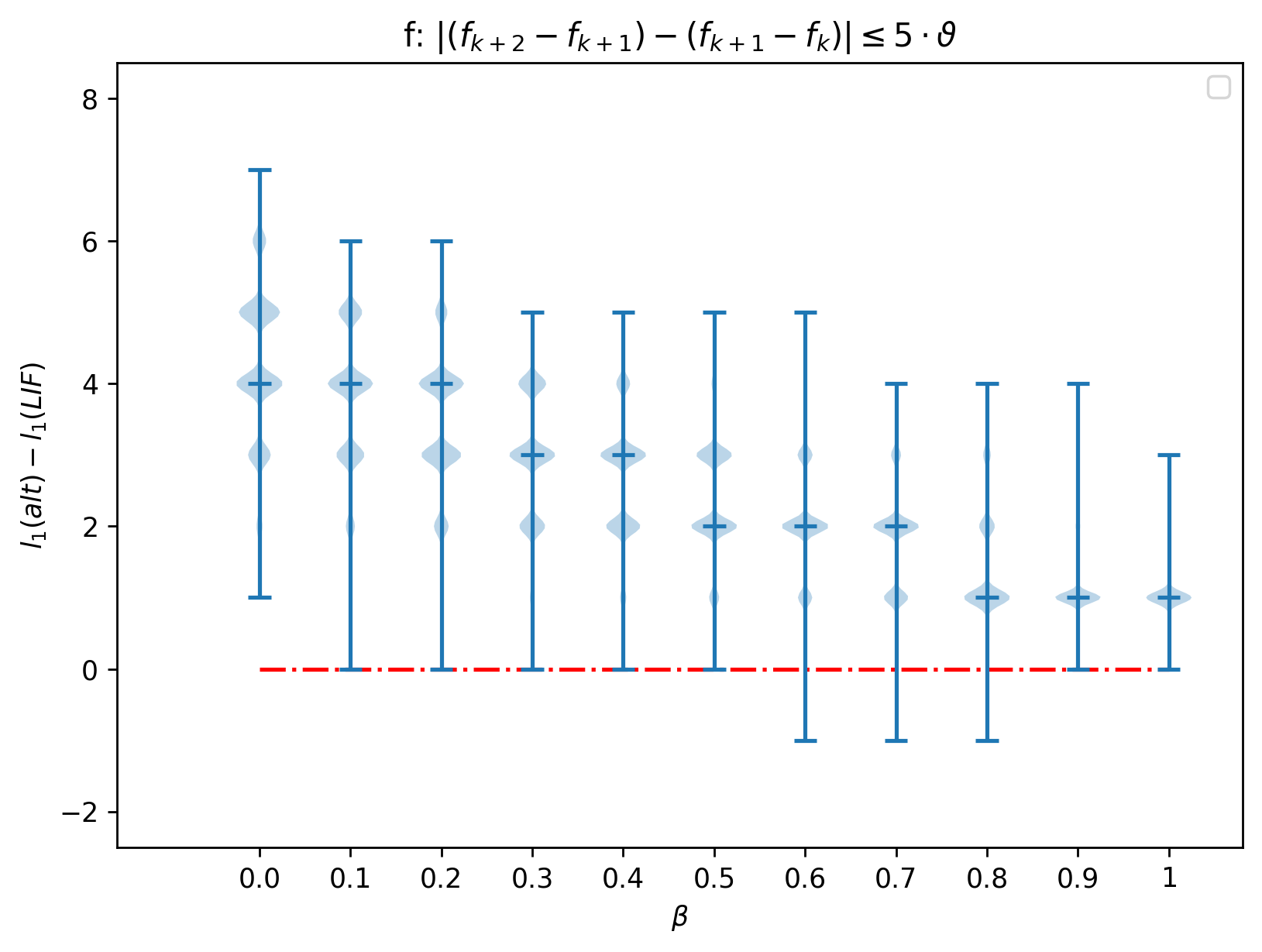}	
		\caption{Continuation of Fig.~\ref{fig:BoxPlot3}. 
		Left: $K= 1$, right: $K=5$.
		}
			\label{fig:BoxPlot3a}
\end{figure}

\section{Conclusion}
\label{s:Conclusion}
As main contribution of this paper we introduced a mathematical framework that allows us
to analyze Leaky Integrate-and-Fire (LIF) as a quantization operator based 
on the Alexiewicz norm. 
We showed how this approach can be used to derive 
sparsity bounds for analog-to-spike conversion based on LIF and how it
distinguishes in terms of an extremal sparsity property.
Our analysis shows in a rigorous mathematical way that information encoding using Threshold-Based Representation (TBR) is intrinsically linked with an analog-to-digital conversion that guarantees sparsity by design, which is of great importance for energy efficiency in edge computing applications. 
As a corollary we show that, thanks to the theorem of Rademacher, the same way of reasoning can also be applied on Send-on-Delta (SOD) by applying
our integration-based approach on the derivative of a Lipschitz continuous input signal. 
Consequently, for SOD we also obtain a law of maximal sparsity where the Alexiewicz norm mutates into the familiar maximum norm.  
In our future research we will exploit these findings for a theory of spike-based signal processing and its connection to spiking neural networks.  

\section*{Acknowledgments}
This work was supported by (1) the ’University SAL Labs’ initiative of Silicon Austria Labs (SAL) and its Austrian partner universities for applied fundamental research for electronic based systems, (2) the COMET Programme via SCCH funded by the Austrian ministries BMK, BMDW, and the State of UpperAustria, and (3) the Horizon Europe Programme via the NeuroSoC project with Europe Grant Agreement number 101070634.

\section*{Appendix A: Proof of upper bound of Theorem~\ref{eq:l1inequ} for general $f$}
We give the proof for $f$ being integrable, bounded and superimposed by a finite number of Dirac pulses. 
Since there are only a finite number of Dirac impulses and the spikes are discrete we can 
choose a sufficiently small $\Delta t>0$ such that all open intervals $(n \Delta, (n+1) \Delta)$ contain at most a single spike 
and at most a single Dirac impulse and such that each spike, resp. each Dirac impulse, falls into some of these intervals. 
Consider now the refined partition resulting from $(k\,\Delta)_k$ and $(t_k)_k$ and let denote it by $(\tau_k)_k$.
Let us set $F_{\tau_k} := \int_{(\tau_{k-1}, \tau_k]} e^{-\alpha(\tau_k - t)} f(t)dt$ and note that
$|F_{\tau_k} | \leq \|f|_{(\tau_{k-1}, \tau_k]}\|_1$. Further, set $\beta_k := e^{-\alpha\, (\tau_k - \tau_{k-1})}$, then we get
\begin{eqnarray}
 &  &\|\mbox{LIF}_{\alpha, \vartheta}(f)\|_1 = \sum_k |s_{t_k}| = \sum_k |s_{\tau_k}|  \nonumber \\
					& = &  \sum_k |q_{\vartheta}(F_{\tau_k} + \beta_k\, (z_{\tau_{k-1}} - s_{\tau_{k-1}}))| \nonumber \\
				& = & 
								\sum_k |F_{\tau_k} + \beta_k\, (z_{\tau_{k-1}} - s_{\tau_{k-1}})| -  |F_{\tau_k} + \beta_k\, (z_{\tau_{k-1}} - s_{\tau_{k-1}}) \nonumber \\
				& & 	\,\,\,\,\, - q_{\vartheta}(F_{\tau_k} + \beta_k\, (z_{\tau_{k-1}} - s_{\tau_{k-1}}))|\nonumber \\
				& \leq & 
				\sum_k |F_{\tau_k}| + \beta_k\, |z_{\tau_{k-1}} - s_{\tau_{k-1}}| -  
				\beta_{k+1}\,|\underbrace{F_{\tau_k} + \beta_k\, (z_{\tau_{k-1}} - s_{\tau_{k-1}})}_{z_{\tau_k}} - \underbrace{q_{\vartheta}(z_{\tau_k})}_{s_{\tau_k}}|\nonumber \\
				& \leq & \sum_k |F_{\tau_k}| \leq \|f\|_1. \nonumber
\end{eqnarray}

\section*{Appendix B: Computing the lower bound of Theorem~\ref{eq:l1inequ}}
The $l_1$-distance of the Alexiewicz ball to the origin leads to a convex $l_1$-minimization problem constrained on a polyhedral set.
In contrast to general $l_1$-minimization problems which are tackled either on the basis of simplex or interior point approaches  
(\cite{barrodale1966algorithms,Donoho2008FastSO,Khodabandeh2011}), the problem in our setting can be resolved explicitly by a simple recursion as outlined next.

Because of (\ref{eq:ball}) we get
\begin{eqnarray}
\label{eq:infBall}
 & \|\mathring{B}_{\alpha, \vartheta}(f)\|_1 \nonumber \\
							& =  \inf \{ \|p\|_1: p \in f + A^{-1}_{\alpha} \circ \mathring{B}_{\infty, \vartheta}(0)\} \nonumber \\
							& =  \inf \{ \|f + A^{-1}_{\alpha}\, c\|_1: |c_k| < \vartheta\} \nonumber \\
							& =  \inf \{ |f_1 + c_1| + |f_2 - \beta\, c_1 + c_2| + \ldots + |f_n - \beta\, c_{n-1} + c_n|: |c_k| < \vartheta\},  
\end{eqnarray}
where $\beta := e^{-\alpha}$.
To solve the minimization problem (\ref{eq:infBall}), first of all, let us consider Lemma~\ref{lem:min}.

\begin{lemma}
\label{lem:min}
For any $a, b \in \mathbb{R}$, $\vartheta>0$ and $\beta \in [0,1]$ we have

\begin{equation}
\label{eq:min}
|a + x^*| +  |b -  \beta \, x^* + y^*| = \min\{|a + x| +  |b -  \beta \, x + y|: x,y \in [-\vartheta, \vartheta] \}, 
\end{equation}

where $x^* := -\mbox{sgn}(a)\, \min\{|a|, \vartheta\}$ and $y^* := -\mbox{sgn}(b-\beta\, x^*)\, \min\{|b-\beta\, x^*|, \vartheta\}$.
For $\beta <1$ the minimum is uniquely determined.
\end{lemma}

\begin{proof}
Consider the functions $p(x):= |a + x|$ and $q_y(x):= |b - \beta\, x + y|$ on the interval $[-\vartheta, \vartheta]$.
Since the absolute value of the slope of $q_y$ is $\beta$ which is  less or equal the modulus of the  slope of $p$ 
the minimum of $p$ is also a minimum of $p +  q_y$. Since $x^* := -\mbox{sgn}(a)\min\{|a|, \vartheta\}$ minimizes $p$, 
it also minimizes $p +  q_y$ for any $y$. Fixing $x = x^*$, therefore $y = y^* := -\mbox{sgn}(b-\beta\, x^*)\, \min\{|b-\beta\, x^*|, \vartheta\}$
minimizes $p +  q_y$ over $x, y \in [-\vartheta, \vartheta]$. The uniqueness property follows the fact that 
the modulus of the slope in the neighborhood of the minimum is not vanishing.
See Fig.~\ref{fig:minLemma} for an illustration. 
\begin{figure}[ht]
	\centering
	\includegraphics[width=0.35\linewidth]{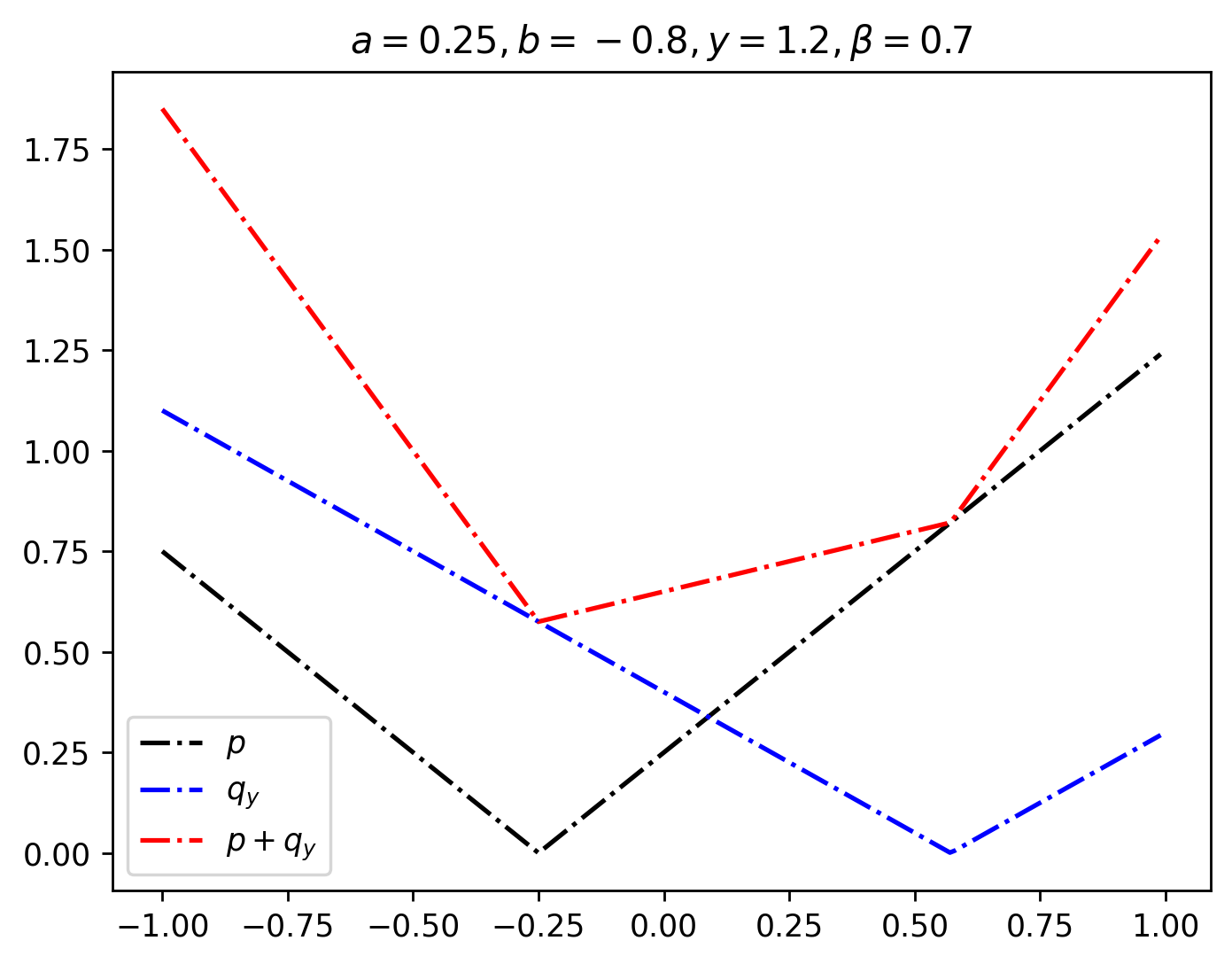}
		\caption{Illustration of Lemma~\ref{lem:min} that the minimizer of $p(x)=|a + x|$ also determines the minimizer of 
		$p + q_y$, where $q_y(x):= |b - \beta\, x + y|$ and $\beta \leq 1$.
		}
			\label{fig:minLemma}
\end{figure}
\end{proof}

\begin{lemma}
\label{lem:minl1}
For any $f \in \mathbb{R}^n$, $\vartheta>0$ and $\beta \in [0,1]$ 
\begin{eqnarray}
\label{eq:p}
p     & := & (f_1 + c_1^*, f_2 - \beta\, c_1^* + c_2^*,\ldots, f_n - \beta\, c_{n-1}^* + c_n^*)^T \in \overline{B}_{\alpha, \vartheta}(f), \\
c_1^*   & := & -\mbox{sgn}(f_1)\, \min\{|f_1|, \vartheta\} \nonumber \\
c_{k+1}^*   & := & -\mbox{sgn}(f_{k+1}- \beta\, c_k^*)\, \min\{|f_{k+1}- \beta\, c_k^*|, \vartheta\} \nonumber 
\end{eqnarray}
is a closest $l_1$-point of the closed Alexiewicz ball $\overline{B}_{\alpha, \vartheta}(f)$ to the origin.
For $\beta <1$ all other points $\widetilde{p}\neq p, \widetilde{p} \in \overline{B}_{\alpha, \vartheta}(f)$ are strictly more distant.
\end{lemma}

\begin{proof}
The proof is based on induction by applying Lemma~\ref{lem:min} on (\ref{eq:infBall}).
Indeed, by fixing the coefficients $c_k$ up to $n-1$, this way setting $R_{n_1} := |f_1 + c_1| + |f_2 - \beta\, c_1 + c_2| + \ldots + |f_{n-1} - \beta\, c_{n-2} + c_{n-1}|$,
we again get the structure of Lemma~\ref{lem:min}.
Therefore a minimizer of the last line of Equ. (\ref{eq:infBall}) over $|c_k|<\vartheta$ is  obtained by successively applying 
Lemma~\ref{lem:min} on 
$$
R_{n-1} + |\underbrace{(f_n - \beta\, c_{n-1})}_{a} + x| + |\underbrace{f_{n+1}}_{b} - \beta\, x + y|.
$$
\end{proof}

\section*{Appendix C: Proof of Theorem~\ref{th:sparsenessIF}}
\begin{proof}
Without loss of generality we may assume $\vartheta = 1$.  First we restrict to the discrete-time case 
$f(t) = \sum_{k=0}^N f_k\, \delta(t-k)$, $f_k \in \mathbb{R}$, to demonstrate the path of the proof.

Equ. (\ref{eq:quantizationForm}) implies that the spike train $s := \mbox{IF}_1(f) \in \mathbb{S}_1$ obtained by 
integrate-and-fire is an admissible solution satisfying the condition $s \in  \mathring{B}^A_1(f) \cap \mathbb{S}_{1}$ 
stated in Equ.~(\ref{eq:sparsenessProp}).

Now we have to show that $s = \mbox{IF}_1(f) \in \mathring{B}^A_1(f) \cap \mathbb{S}_{1} $ is at least as sparse as any other 
spike train $s^* \in \mathring{B}^A_1(f) \cap \mathbb{S}_{1}$, that is we have to show that
\begin{equation}
\label{eq:IFsparsi}
\|s\|_1 \leq \|s^*\|_1.
\end{equation}

Therefore, first note that the Alexiewicz norm $\|s - s^*\|_A$ of the difference 
$s - s^*$ is the maximum of partial sums of differences of integers ($s_k, s_k^* \in \mathbb{Z}$), hence $\|s - s^*\|_A \in \mathbb{N}_0$.
Now, from the triangle inequality of a norm and 
the condition $s, s^* \in \mathring{B}^A_1(f) \cap \mathbb{S}_{1}$ it follows that $\|s - s^*\|_A \leq \|s - f\|_A + \|s^* - f\|_A < 2$, 
hence  
\begin{equation}
\label{eq:assIFs*1}
\|s - s^*\|_A \in \{0,1\}.
\end{equation}

In the discrete-time case we denote $f = \sum_k f_k \, \delta(t - k)$,  $s = \sum_k s_k \, \delta(t - k)$ and 
$s^* = \sum_k s^*_k \, \delta(t - k)$. Let denote by $i_k$ the trigger points of $s$, i.e.,  $s_{i_k} \neq 0$, and by 
$i^*_k$ the points of $s^*$ with $s^*_{i_k} \neq 0$.
The proof is based on induction by showing that a maximally sparse solution $s^*$ satisfies $\sum_{j=0}^{i_k} |s_i| = \sum_{i=0}^{i_k} |s^*_i|$ for all $k$.
Note that this property entails $\sum_{j=0}^{n} |s_i| \leq  \sum_{j=0}^{n} |s^*_i| $ for all $s^* \in \mathring{B}^A_1(f) \cap \mathbb{S}_{1}$.
\\
\\
\noindent{\it Initial Induction Step.}
Note that $\|f|_{[0, i_1)}\|_A < 1$ and that $\|f|_{[0, t_1]}\|_A\geq 1$ which causes IF to trigger a spike with $s_{i_1} \in \mathbb{Z}$ so that the difference 
between the signal $f$ and the spike train $s$ on the interval $[0, t_1]$ becomes below threshold again, 
i.e., $\|(f-s)|_{[0, t_1]}\|_A < 1$. This implies that $i^*_1 \leq i_1 $, as the contrary $i^*_1 > i_1$ would imply 
$\|f - s^*\|_A \geq \|(f- s^*)|_{[0, i_1]}\|_A\geq 1$, contradicting the strict-within-ball assumption 
$s^* \in \mathring{B}^A_1(f) \cap \mathbb{S}_{1}$ of $s^*$, and, by this showing 
\begin{equation}
\label{eq:indStepDisc}
\sum_{i=0}^{i_1} |s_i| = \sum_{i=0}^{i_1} |s^*_i|.
\end{equation}
\vspace{1cm}
\noindent{\it  Induction Step for $k=n+1$.}
Now assume 
\begin{equation}
\label{eq:indStepDiscN}
\forall k \leq n: \sum_{j=0}^{i_k} |s_i| = \sum_{i=0}^{i_k} |s^*_i|
\end{equation}
and suppose the contrary for $k=n+1$, that is
\begin{equation}
\label{eq:indStepDiscContra}
\sum_{j=0}^{i_{n+1}} |s_i| > \sum_{i=0}^{i_{n+1}} |s^*_i|.
\end{equation}
The strict-within-ball assumption $s, s^* \in \mathring{B}^A_1(f) \cap \mathbb{S}_{1}$ implies
\begin{equation}
\label{eq:indStepDisc1}
\left|F_k - \sum_{i=0}^{k} s_i \right| < 1, \,\, \left|F_k - \sum_{i=0}^{k} s^*_i \right| < 1,
\end{equation}
for all $k$, where $F_k = \sum_{i=0}^{k} f_i$ and  $\sum_{i=0}^{k} s_i, \sum_{i=0}^{k} s^*_i \in \mathbb{Z}$.
Due to Theorem~\ref{th:IFDecomposition} 
we obtain $\sum_{i=0}^{i_n} s_i = q(F_{i_n}) \in \mathbb{Z}$, 
and because of the second inequality of (\ref{eq:indStepDisc1})
we have two possibilities 
\begin{equation}
\label{eq:pm}
\sum_{i=0}^{i_{n}} s^*_i \in \{q(a), q(a) \pm 1\}.
\end{equation}
depending on whether $q(a)\geq 0$ ($\pm 1$ in Equ.~(\ref{eq:pm}) becomes $+1$) or $q(a) < 0$ ($\pm 1$ in Equ.~(\ref{eq:pm}) becomes $-1$).

This way we get $\sum_{i=0}^{i_{n}} s_i = \sum_{i=0}^{i_{n}} s^*_i +d$ with $d \in \{-1, 0,1\}$.
By taking the induction assumption (\ref{eq:indStepDiscN}) into account and splitting the positive and the negative spikes into two groups we obtain
\begin{eqnarray}
\sum_{s_i>0} |s_i| -  \sum_{s_i<0} |s_i|  & \stackrel{(\ref{eq:pm})}{=} & \sum_{s^*_i>0} |s^*_i| -  \sum_{s^*_i<0} |s^*_i| + d, \nonumber\\
\sum_{s_i>0} |s_i| +  \sum_{s_i<0} |s_i|  & \stackrel{(\ref{eq:indStepDiscN})}{=} & \sum_{s^*_i>0} |s^*_i| +  \sum_{s^*_i<0} |s^*_i|, \nonumber
\end{eqnarray}
for $i \in \{1, \ldots, i_n\}$, and, therefore, by adding both lines, we finally get 
$\sum_{s_i>0} |s_i| - \sum_{s^*_i>0} |s^*_i| = d/2 \in \mathbb{Z}$, which can only be valid in the case of $d=0$.
Consequently we conclude from (\ref{eq:pm}) to get $\sum_{i=0}^{i_n} s_i = \sum_{i=0}^{i_n} s^*_i$.
Now let us consider the interval from $i_n$ to the next triggering point $i_{n+1}$. 
For $s = \mbox{IF}_1(f)$ we get 
\begin{equation}
\label{eq:sIndDisc}
\left|\sum_{i=0}^{i_{n}} f_i- \sum_{i=0}^{i_{n}} s_i + \sum_{i = i_{n}+1}^{i_{n+1}} f_i \right| \geq 
\left|s_{i_{n+1}}\right| \stackrel[]{(\ref{eq:indStepDiscContra})}{>} \sum_{i=i_n+1}^{i_{n+1}} |s^*_i|,
\end{equation}
while for $s^*$ we have
\begin{equation}
\label{eq:s*IndDisc}
\left| \left|\sum_{i=0}^{i_{n+1}} f_i - \sum_{i=0}^{i_{n}} s^*_i \right| -\left|\sum_{i=i_n+1}^{i_{n+1}} s^*_i\right| \right|
\leq 
\left| \sum_{i=0}^{i_{n+1}} f_i - \sum_{i=0}^{i_{n+1}} s^*_i\right| < 1.
\end{equation}
Therefore, because of $\sum_{i=0}^{i_{n}} s^*_i = \sum_{i=0}^{i_{n}} s_i$ and (\ref{eq:sIndDisc}), we get
the contradiction
\begin{equation}
\label{eq:s*IndDisc1}
\left|\sum_{i=0}^{i_{n+1}} f_i - \sum_{i=0}^{i_{n}} s_i  \right| < 
1 + \left|\sum_{i=i_n+1}^{i_n+1} s^*_i \right| \leq 1+ \sum_{i=i_n+1}^{i_{n+1}} |s^*_i| \leq \left|\sum_{i=0}^{i_{n+1}} f_i - \sum_{i=0}^{i_{n}} s_i  \right|,
\end{equation}
which falsifies the assumption of Equ.~(\ref{eq:indStepDiscContra}) and proves 
$\sum_{j=0}^{i_{n+1}} |s_i| = \sum_{i=0}^{i_{n+1}} |s^*_i|$. 
This means, there is no other spike train $s^* \in \mathring{B}^A_1(f) \cap \mathbb{S}_{1}$ that is strictly sparser than $\mbox{IF}_1(f)$.
Note that the proof in continuous time works analogously as the induction iterates over the discrete number of triggering points.
\end{proof}




\end{document}